\documentclass[journal,comsoc]{IEEEtran}
\IEEEoverridecommandlockouts
\usepackage{cite}
\usepackage{amsmath,amssymb,amsfonts}
\usepackage{graphicx}
\usepackage{textcomp}
\usepackage{xcolor}
\usepackage{cleveref}
\def\BibTeX{{\rm B\kern-.05em{\sc i\kern-.025em b}\kern-.08em T\kern-.1667em\lower.7ex\hbox{E}\kern-.125emX}}
\usepackage{multirow}
\usepackage{booktabs}
\usepackage{amsfonts}
\usepackage{amsmath}
\usepackage{amsmath,cases}
\usepackage{amssymb}
\usepackage{blindtext, graphicx}
\usepackage{cite}
\usepackage{epsfig}
\usepackage{url}
\usepackage{algorithm}
\usepackage{algorithmicx, algpseudocode}
\usepackage{epstopdf}
\usepackage{color}
\usepackage[tight]{subfigure}
\usepackage{mathrsfs}
\usepackage[justification=centering]{caption}
\usepackage{float}
\usepackage{cuted}
\usepackage{lipsum}
\usepackage{diagbox}
\usepackage{graphicx,tabularx}
\usepackage{lipsum}
\usepackage{enumitem}
\usepackage{academicons}
\usepackage{xcolor}

\newcommand{\orcid}[1]{\href{https://orcid.org/#1}{\textcolor[HTML]{A6CE39}{\aiOrcid}}}

\newtheorem{theorem}{\pmb{Theorem}}

\newcommand{\pmw}{\pmb{w}}

\newcommand{\pmW}{\pmb{W}}
\newcommand{\sigk}{\sigma_{k}^2}
\newcommand{\sige}{\sigma_{e}^2}
\newcommand{\pmt}{\pmb{\theta}}

\newcommand{\btheta}{\pmb{\theta}}
\newcommand{\thetak}{\pmb{\theta}^{(\iota)}}
\newcommand{\thetako}{\pmb{\theta}^{(\iota+1)}}

\newtheorem{mylem}{\textbf{Lemma}}

\newcommand{\la}{\langle}
\newcommand{\ra}{\rangle}

\newcommand{\Prf}{{\it Proof: \ }}

\newcommand{\clC}{{\cal C}}

\newcommand{\clR}{{\pmb{\cal R}}}

\newcommand{\wko}{\pmb{w}^{(\iota+1)}}

\newcommand{\wk}{\pmb{w}^{(\iota)}}

\newcommand{\bbC}{\mathbb{C}}

\newcommand{\upsilonk}{\upsilon^{(\iota)}}

\allowdisplaybreaks
\begin{document}
	\title{Secrecy Sum-Rate Maximization in Finite Blocklength IRS-aided Systems With Perfect and Imperfect CSI}
	\author{Monir Abughalwa,~\IEEEmembership{Student Member,~IEEE}, Nguyen Van Huynh,~\IEEEmembership{Member,~IEEE}, Diep N. Nguyen,~\IEEEmembership{Senior Member,~IEEE}, Dinh Thai Hoang,~\IEEEmembership{Senior Member,~IEEE}, and Eryk Dutkiewicz,~\IEEEmembership{Senior Member,~IEEE}
		\thanks{Monir Abughalwa, Diep N. Nguyen, Dinh Thai Hoang, and Eryk Dutkiewicz are with School of Electrical and Data Engineering, University of Technology Sydney, Sydney, Australia (e-mail: monir.abughalwa@student.uts.edu.au; diep.nguyen@uts.edu.au; hoang.dinh@uts.edu.au; eryk.dutkiewicz@uts.edu.au).}
		
		\thanks{Nguyen Van Huynh is with the Department of Electrical Engineering and Electronics, University of Liverpool, Liverpool L69 3GJ, United Kingdom (e-mail: huynh.nguyen@liverpool.ac.uk)}
		%\thanks{Ming Zeng is with Department of Electrical Engineering and Computer Engineering, Universite Laval, Quebec, QC G1V 0A6, Canada (e-mail: ming.zeng@gel.ulaval.ca).}
		%\thanks{Quoc-Viet Pham is with Department of School of Computer Science and Statistics, Trinity College Dublin, Dublin 2, D02 PN40, Ireland (e-mail: viet.pham@tcd.ie).}
		%	\thanks{Van-Dinh Nguyen is with College of Engineering and Computer Science, VinUniversity, Vinhomes Ocean Park, Hanoi, Vietnam (e-mail: dinh.nv2@vinuni.edu.vn).}
	}

	\maketitle
\begin{abstract}
The rapid growth of the Internet of Things (IoT) requires efficient and secure communication technologies to enable ultra-reliable low-latency (URLLC) communication applications. Intelligent reflecting surfaces (IRS) have emerged as a promising solution to enhance IoT network secrecy performance by improving signal quality for legitimate devices (Bob), while thwarting the eavesdropper (Eve) interception. However, ensuring secrecy across multiple users, given their diverse channels and locations, is challenging, especially under finite blocklength (FBR) constraints, which are common in IoT networks. This paper investigates the secrecy performance of IRS-aided URLLC systems under three channel state information (CSI) scenarios: perfect CSI, imperfect CSI, and unknown eavesdropper's CSI. We first formulate a non-convex optimization problem to maximize the system's sum secrecy rate (SSR) by jointly optimizing the transmitter's beamforming and the IRS's passive reflective elements, while maintaining FBR-related latency and transmission duration constraints. For imperfect CSI, the semi-infinite uncertainty constraints are transformed into finite linear matrix inequalities (LMIs) via a successive convex approximation (SCA)-based approach, and the proposed algorithm is proven to converge to a locally optimal solution with low computational complexity. Furthermore, when Eve's CSI is unavailable, we formulate a power minimization problem that ensures each user's quality-of-service (QoS) requirement, while the residual transmit power is allocated as artificial noise (AN) to degrade Eve's signal-to-interference-plus-noise ratio (SINR). This problem is efficiently solved using SCA approach and penalty convex-concave procedure (PCCP) framework, ensuring secrecy even under stringent FBR conditions. 
Simulations validate the effectiveness of the proposed algorithms under all CSI scenarios and highlight the impact of reflected-channel uncertainty.
\end{abstract}

	\begin{IEEEkeywords}
		Finite blocklength regime, ultra-reliable and low latency communication,  intelligent reflecting surfaces (IRS), and sum secrecy rate (SSR), Artificial Noise (AN).
	\end{IEEEkeywords}

	\section{Introduction} 
	Intelligent reflecting surface (IRS) has emerged as a key enabler for IoT networks, garnering significant interest due to its ability to enhance/tailor the passive radio environment. A typical IRS consists of passive reflective elements (PREs) equipped with phase shift controllers, allowing it to intelligently manipulate reflected signals \cite{nguyen2022leveraging}. 
    By optimizing signal reflections, IRSs can enhance signal reception by creating favorable multi-path components at the receivers (Rx/Bob). Thanks to their cost-effective and flexible design, IRSs can be conveniently deployed in various urban scenarios, such as on lampposts, building facades, or moderate-height structures to overcome blockage and extend coverage \cite{9326394}. Although IRSs mounted on high-rise buildings are also feasible, their deployment is particularly advantageous for serving IoT devices located at comparable altitudes (e.g., sensors or access points on nearby buildings). These versatile deployment options highlight the practical potential of IRSs for improving link reliability in dense urban environments where line-of-sight (LoS) channels between the transmitter (Tx/Alice) and the Rx are frequently obstructed~\cite{abughalwa2022finite}.
    Particularly, for the Internet of Things (IoT) devices that have limited computing capability and battery/energy, the IRS has gained paramount attention, aiming to enhance both spectral and energy efficiency \cite{9896755}. For example, the authors in \cite{9270605} studied the IRS deployment impact on IoT systems with mobile edge computing. Numerical results revealed that the IRS can significantly enhance the system's computational performance. In \cite{10700955}, a satellite-terrestrial integrated network with an IRS was investigated to enhance physical-layer security against eavesdroppers. The work formulated an SR maximization problem subject to IRS energy-harvesting and total transmit-power constraints.
    However, incorporating IRS into the finite blocklength (short packet) regime/communication (FBR) that is used in ultra-reliable low-latency (URLLC) applications is more challenging than the traditional long blocklength regime/communication (LBR). This is because of the FBR's strict latency and transmission-duration constraints.
	\subsection{Related Works and Motivations}
	Another potential application of IRSs is to enhance the security/privacy of users by purposely manipulating reflected signals from the transmitter to facilitate the signal reception at legitimate users while maximizing the multi-user interference/degrading the signals at potential eavesdroppers. For example, a multi-input-single-output IRS-aided system with the presence of multiple eavesdroppers was considered in \cite{hong2020robust}, where the transmitter introduced artificial noise as a countermeasure to enhance user security. The authors in \cite{9726800} investigated IRS-aided hybrid satellite-terrestrial relay networks to address blocked user links and proposed an alternating optimization scheme for joint beamforming and IRS phase-shift design, aiming to minimize transmit power under rate constraints. In \cite{10980348}, the authors surveyed space-air-ground integrated networks (SAGIN), outlining key security challenges and introducing wireless endogenous security based on physical-layer attributes. Kang et al. in \cite{10520169} similarly reviewed SAGIN architectures, highlighted security vulnerabilities in heterogeneous environments, and discussed hybrid frequency-power-spatial domain strategies for enhancing secure and efficient network access.
    In \cite{hao2022securing}, the authors studied a cell-free IRS-aided network with the aim of enhancing the users' secrecy by maximizing the weighted sum of the users' secrecy rate (SR) by jointly optimizing the beamforming vector and the IRS's PREs. The alternating optimization (AO) algorithm was used to decouple the beamforming factor and the IRS's PREs. Then, the resulting problem was tackled using semi-definite relaxation (SDR) and continuous convex approximation.  When the channel state information (CSI) from the IRS to the users/receivers is unknown or imperfect, the authors in \cite{10256584} studied an IRS-aided multi-user system, where they formulated a sum secrecy rate (SSR) maximization problem with the eavesdropper's channel partially known to the Tx/Alice. The IRS was deployed to enhance the uplink transmission and the downlink energy transfer. Then, the authors optimized the SSR by jointly optimizing the transmit beamform vectors, the downlink/uplink time allocation, and the energy transmit covariance matrix.
    
	In parallel, URLLC relying on FBR, have been envisioned as key technologies to serve critical missions of IoT. %Since FBR often requires a stricter design approach than the LBR systems, maintaining high-reliability communication is more challenging due to the lower channel coding gain~\cite{durisi2016toward}. Moreover, 
    In URLLC applications such as intelligent transportation, leakage information could expose the user's location or identity, due to the inherently private nature of the IoT transmitted information \cite{feng2021reliable}. As a potential solution to secure URLLC, the use of IRS has recently attracted paramount interest. This stems from the fact that IRS links exhibit low latency, which is beneficial for the FBR strict requirements \cite{10904325}. For example, Zhao et al. studied the information freshness in a single-user FBR-IRS-aided system with the presence of an eavesdropper in \cite{10904325}. The authors first derived a closed-form expression for the upper bound secrecy outage probability under statistical CSI. Then, the single user's secrecy outage is minimized by jointly optimizing the blocklength and the IRS PREs. Differently, the authors in \cite{9833996} investigated the problem of transmission power minimization in an IRS-aided system under the FBR regime with the presence of multiple eavesdroppers, by jointly optimizing the beamforming vector, the AN, and the IRS PREs. In \cite{10501603}, the authors studied the SSR maximization problem for an FBR-IRS-aided system by jointly optimizing the beamforming vector, the IRS PREs, and the blocklength. The optimization problem was tackled using Riemannian conjugate gradient based algorithm. However, the assumption that Tx/Alice can obtain Eve's CSI, as stated before, is not practical. In addition, most of the existing studies focused only on single-user systems, and thus their applications are limited.
    
    When only partial or erroneous CSI (from the IRS to the users and the eavesdropper) is available, the problem is even more challenging due to the semi-infinite constraints arising from CSI uncertainty, which are embedded within the dispersion term of the FBR-SR expression, which contains both the inverse Q-function and the square-root of the channel dispersion. This structure makes the decoding error probability highly sensitive to channel estimation errors \cite{7529226}, while the square-root term further intensifies the nonlinearity. To address these challenges, we first introduce slack variables to decouple the transmit beamforming vectors and the IRS PREs within the objective function. Then, we employ the successive convex approximation (SCA) framework \cite{9197675} together with the $\mathcal{S}$-procedure \cite{boyd1994linear} to transform the semi-infinite constraints into linear matrix inequalities (LMIs). Although this transformation yields LMIs, they remain nonconvex. To overcome this, we apply a first-order Taylor expansion to linearize the nonconvex parts, resulting in a convex approximation that can be efficiently solved at each iteration. Then, we propose a penalty convex-concave procedure (PCCP) \cite{9525400} to tackle the unit modulus constraint (UMC) of the IRS.
    When Alice can't obtain Eve's CSI, we propose a power minimization problem to ensure the users' data rate meets a minimum quality of service (QoS) constraint, while maintain the FBR constraints, namely, the transmission duration and latency. Once the transmission  power is minimized, the remaining residual power is allocated as AN to degrade Eve's signal-to-interference-plus-noise ratio (SINR). The proposed problem is tackled by jointly optimizing the beamforming vectors and the IRS's PREs. We use SCA and $S$-procedure to optimize the beamforming vector. Then, we use a PCCP algorithm to optimize the IRS's PREs, while ensuring the FBR's constraints and the imperfect IRS to the users' CSI. Unlike traditional LBR power minimization problems and existing research \cite{9146177,9764813}, our method explicitly incorporates the FBR constraints and addresses challenges arising from imperfect channel state information.
    Extensive simulations with practical settings show that maximizing the SSR of all the users can achieve a better chance of ensuring system secure communications (subject to the location of the eavesdropper) under the FBR constraints, even under imperfect CSI conditions.
    
    \subsection{Contributions}
	The main contributions are summarized as follows: 
	\begin{itemize}
		\item We tackle the secrecy maximization problem of the URLLC users in the FBR-IRS-aided system, in which the IRS enhances the users' SINR and thwarts Eve's SINR to ensure the system's SR. To this end, we maximize the SSR by jointly optimizing the beamforming vectors and the IRS's PREs, while satisfying the FBR latency and transmission duration constraints. 
        \item When CSI is available, we solve the above non-convex problem by linearizing its objective function with tractable approximation functions, leading to a computationally efficient algorithm. To tackle UMC, unlike traditional methods, i.e., conventional SDR, we directly optimize the PREs arguments to provide a low-complexity solution that can scale with large IRSs. The resulting solution is proved to converge to a locally optimal solution of the original non-convex problem.
		\item In the case of imperfect CSI, we transform the semi-infinite constraints imposed by the imperfect CSI constraint into finite LMIs by leveraging the SCA approach for the transformation. We then prove that our proposed algorithm converges to a locally optimal solution with low computational complexity thanks to our closed-form linearization approach. This makes the solution scalable for large IRS deployments.
        \item In scenarios where Eve's CSI is not accessible, we tackle a power minimization problem within the FBR framework, ensuring that every user's data rate meets a specified QoS threshold. The remaining transmit power is deliberately used as AN to decrease the eavesdropper's SINR. The UMC is handled using the PCCP framework while maintaining the QoS criteria. The proposed algorithm is demonstrated to converge to a locally optimal solution, effectively ensuring system-level secrecy under strict FBR conditions.
		\item Finally, we perform extensive simulations with practical settings. The simulation results show that our approach can ensure secure communication for all users while satisfying FBR constraints even under imperfect CSI. In addition, the impact of the imperfect CSI of the reflected channel on the system is also evaluated.
	\end{itemize}
	The remainder of this paper is organized as follows. The system model is discussed in Section \ref{sec2}. Then, Sections \ref{sec3} and \ref{sec4} present the problem and corresponding solutions under perfect and imperfect CSI, respectively. Section~\ref{sec5} further extends the formulation to the scenario where the eavesdropper's CSI is unavailable. The extensive simulations and discussion are provided in Section \ref{sec6}. Finally, Section \ref{conc} concludes the paper.
	
	\emph{Notations:} This paper uses the following notations. The bold letters denote vectors and matrices. $\pmb{I}_M$ denotes the identity matrix of dimension $M$. Diag($n_1,\dots,n_n$) denotes the diagonal matrix with diagonal entries of $\{n_1,\dots,n_n\}$. The symbols $\Re$ and $\mathbb{C}$ represent the real and the complex field, respectively. $\clC(0,\bar{z})$ denotes the circular Gaussian random variable with zero mean and variance $\bar{z}$. For matrices $\mathbf{C}$ and $\mathbf{D}$, $\la \mathbf{C},\mathbf{D} \ra \triangleq \mbox{trace}\left(\mathbf{C}^H \mathbf{D}\right)$. For matrix $\mathbf{A}$, $\Re\left\{\mathbf{A}\right\}$ denotes the real part, $\la \mathbf{A}\ra \triangleq\mbox{trace}(\mathbf{A})$, the symbol $\|\mathbf{A}\|_1$ denotes the 1-norm, $\|\mathbf{A}\|$ denotes the Frobenius norm, $\mathbf{A}^*$ denotes the conjugate, $\mathbf{A}^H$ denotes the Hermitian (conjugate transpose), $\lambda_{\text{max}}\left(\mathbf{A}\right)$ denotes the maximal eigenvalue, $\angle(\mathbf{A})$ denotes its argument, and $\mathbf{A} \succeq 0$ means positive semi-definite.  
	Lastly, the notations used in this paper are shown in Table~\ref{tab:notations}.
\begin{table}[!t]
\caption{List of Notations}
\label{tab:notations}
\centering
\begin{tabular}{>{\bfseries}p{0.18\linewidth} p{0.72\linewidth}}
\toprule
Symbol & Description \\
\midrule
$\mathcal{K}$ & Set of legitimate users \\
$k, e$ & User index and Eve index \\
$M$ & Number of antennas at Alice (base station) \\
$N$ & Number of IRS reflecting elements (PREs) \\
$\pmb{G}_{\text{AR}}$ & Channel from Alice to the IRS \\
$\pmb{u}_{i}$ & Channel from IRS to user/Eve, $i \in \{k,e\}$ \\
$\pmb{h}_{i}(\btheta)$ & Cascaded channel gain from Alice to $i$ \\
$\widehat{\pmb{h}}_{i}(\btheta)$ & Imperfect cascaded channel gain \\
$\pmb{\Phi}$ & IRS reflection matrix \\%defined as $\text{Diag}(e^{j\btheta})$ \\
$\btheta$ & IRS phase shift vector \\
$\widehat{\pmb{u}}_{i}$ & Estimated channel vector \\
$\Delta\pmb{u}_{i}$ & Channel estimation error \\
$\Omega_{i}$ & Radius of the CSI uncertainty region \\
$\omega_i$ & Set of all possible channel estimation errors \\
$\pmw_k$ & Beamforming vector for user $k$ \\
$\sigma_i^2$ & Noise power at receiver $i$ \\
$\widehat{\gamma}_i$ & SINR at user/Eve $i \in \{k,e\}$ \\
$\widehat{C}_{i}(\pmw,\pmt)$ & Achievable rate/eavesdropping rate \\
$\widehat{\mathcal{S}}_{k}^{\mathcal{F}}$ & Finite blocklength secrecy rate (FBR-SR) \\
$\widehat{\mathcal{S}}_{k}^{\mathcal{L}}$ & Long blocklength secrecy rate (LBR-SR) \\
$\mathcal{B}$ & System bandwidth \\
$T$ & Transmission duration \\
$N_t$ & Packet length \\%($N_t = \mathcal{B}T$) \\
$\varsigma_k$ & Decoding error probability of user $k$ \\
$\varsigma_e$ & Information leakage probability at Eve \\
$Q^{-1}(\cdot)$ & Inverse Gaussian Q-function \\
$\xi_i$ & Finite blocklength factor\\% $\frac{Q^{-1}(\varsigma_i)}{\ln(2)\sqrt{N_t}}$ \\
$V_i$ & Channel dispersion factor \\
$P_T$ & Transmit power at Alice \\
$P_{tot}$ & Total available transmit power \\
$P_J$ & Residual power allocated as artificial noise (AN) \\
$\mathbf{V}$ & Artificial noise projection matrix \\
$\mathbf{r}$ & Artificial noise vector \\
\bottomrule
\end{tabular}
\end{table}

	\section{{System Model}} \label{sec2}
    	\begin{figure}[!htb]
		\centering
		\includegraphics[width=0.4\textwidth]{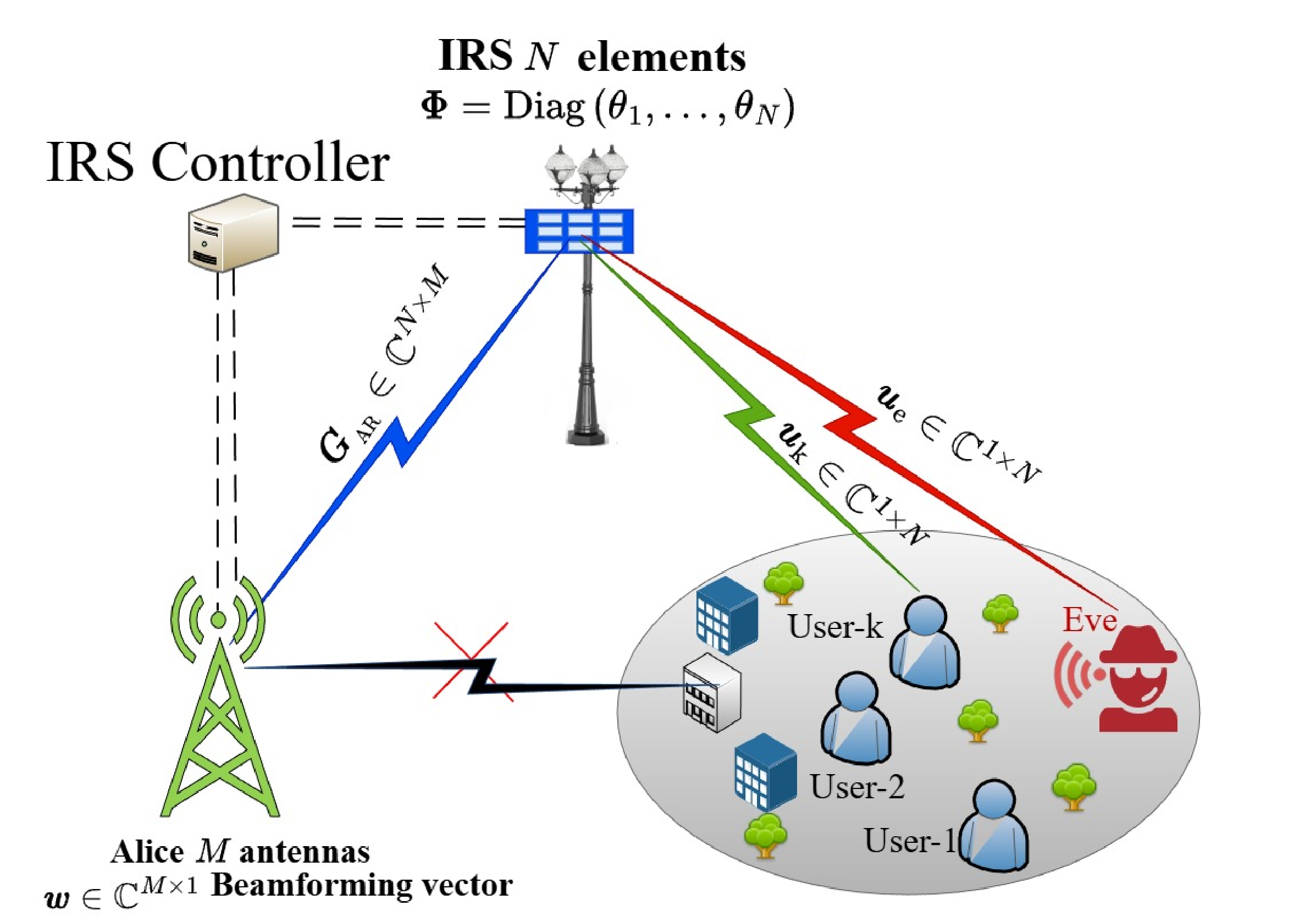}
		\caption{System model.}
		\label{s1}
	\end{figure}
	We consider an FBR-IRS-aided downlink system as depicted in Fig. \ref{s1}, in which an $M$-antenna BS (Alice) transmits confidential information to $\mathcal{K}$ single antenna IoT-users under the presence of a single antenna eavesdropper (Eve). Here, we assume that direct radio links between Alice and users are severely blocked{\footnote{The effectiveness of the IRS in enhancing users' rate/confidentiality is limited when a strong direct link between the transmitter and the receiver exists \cite{9963699}.}}. This scenario is usually the case in highly populated areas with high-rise buildings. An IRS with $N$ PREs is thus deployed (e.g., at the facade of a building) to support the transmission between Alice and the users. Let $k \triangleq \{1,2, \dots, \mathcal{K}\}$ denote the set of legitimate users, and $e$ denote Eve in the system. Let 
	$\pmb{G}_{_\text{AR}}= \sqrt{\beta_{_{\text{{AR}}}}} \tilde{{\pmb{G}}}_{_{\text{{AR}}}}\in \mathbb{C}^{N \times M}$ denote the channel from Alice to the IRS, where $\tilde{{\pmb{G}}}_{_{\text{{AR}}}}$ is modeled by Rician fading, and $\sqrt{\beta_{_{\text{{AR}}}}}$ is the large scale fading factor of the Alice-to-IRS link. The channel from the IRS to user-$k$ and to Eve is modeled by $\pmb{u}_i=  \sqrt{\beta_{\text{R}i}} \tilde{\pmb{u}}_i \clR^{1/2}_{\text{R}i}\in \mathbb{C}^{1\times N}$, where $i \in \{k,e\}$, $\beta_{\text{R}i}$ is the large-scale fading factor of the IRS-to-$i$ \cite{kammoun2020asymptotic}, $\tilde{\pmb{u}}_i$ is modeled by Rician fading, and $\clR_{\text{R}i} \in \bbC^{N\times N}$ is the IRS elements' spatial correlation matrix \cite{kammoun2020asymptotic}. For the confidential message intended for user-$k$, i.e., $s_k$, the signal received by user-$k$ and Eve corresponding to the intended user, can be respectively expressed by:
	\begin{align} \label{chkE}
		y_i\triangleq \pmb{h}_{i}(\btheta) {\sum}_{k=1}^{\mathcal{K}} \pmw_k s_k +n_i, i\in \{k,e\},
	\end{align}
	where  $\pmb{h}_{i}(\btheta) \in \mathbb{C}^{1 \times M}$ is the cascaded channel gain from Alice to $i \in \{k,e\}$, $\pmw_k \in \mathbb{C}^{M \times 1}$ is the beamforming vector applied for user-$k$,  $\btheta=(\theta_1,\dots, \theta_N)^T\in [0,2\pi)^N$ denotes the phase shift vector of the IRS's PREs, $n_i$ is the zero-mean additive white Gaussian noise (AWGN) with power density $\sigma_i^2$ for $i \in \{k,e\}$. The cascade channel gain $\pmb{h}_{i}(\btheta)$ from Alice to $i \in \{k,e\}$, can be written in terms of the channel gain from Alice to the IRS and the channel gain from the IRS to the user or Eve as follows:
	\begin{align} \label{h1} 
		\pmb{h}_{i}(\btheta)\triangleq {\pmb{u}}_i \pmb{\Phi} \pmb{G}_{_\text{AR}} \triangleq {\pmb{u}}_i {\sum}_{n=1}^{N} \exp(j \theta_n) \Xi_n \pmb{G}_{_\text{AR}},
	\end{align}
	where $\pmb{\Xi}_n$ is an $N \times N$ matrix with all zeros except for its $(n,n)$ entry which is 1, and $\pmb{\Phi} = \text{Diag}(e^{j \btheta})$.

	The CSI from Alice to the IRS can be estimated by calculating the angle of arrival and departure \cite{9110587}. However, the reflected channel's CSI from the IRS to the IoT users is much more challenging to obtain due to the passive nature of the IRS and the mobility/varying nature of the user's environments and locations \cite{8579566}. For Eve, it is also hardly possible to pinpoint its location or its accurate CSI from the IRS. To account for the CSI imperfection, we adopt the bounded CSI model in which the reflected channel from the IRS to the users and Eve can be expressed as \cite{9266086}:
	\begin{subequations} \label{ICSIcha}
		\begin{alignat}{2} 
			&\pmb{u}_{i}&& \triangleq \widehat{\pmb{u}}_{i} + \Delta\pmb{u}_{i}, \forall i \in \{k,e\}, \label{ch2-h1} \\            
			&\omega_i && \triangleq \{\|\Delta\pmb{u}_{i}\|_2 \leq \Omega_{i}\}, \forall i \in \{k,e\},\label{ch2-h2}	
		\end{alignat}
	\end{subequations}
	where $\widehat{\pmb{u}}_{i}$ denotes the (imperfect) estimated channel vector, $\Delta\pmb{u}_{i}$ represents the channel estimation error of the corresponding estimation, $\omega_i$ is a set for all possible channel estimation errors, and $\Omega_{i}$ is the radii of the uncertainty regions as known to Alice.  Hence, \eqref{h1} can be reformulated as:
	\begin{align} \label{chaCSIdef}
		\widehat{\pmb{h}}_{i}(\btheta)\triangleq \left(\widehat{{\pmb{u}}}_i+ \Delta\pmb{u}_{i}\right) \pmb{\Phi} \pmb{G}_{_\text{AR}}.
	\end{align}
	 In the sequel, we deal with perfect/imperfect CSI from the IRS to the users and Eve, as well as where Eve's CSI is unknown. 
    Using different well-established channel estimation methods, such as the anchor-assisted channel estimation approach \cite{9603291}, the users' CSI can be practically obtained by Alice. In this case, the reflected channel from the IRS to the users and Eve can be captured by setting $\Delta\pmb{u}_{i}$ to zero in \eqref{ch2-h1}, i.e., $\pmb{u}_{i}= \widehat{\pmb{u}}_{i}, i \in \{k,e\}$. In the second case, only partial/imperfect CSI $\widehat{\pmb{u}}_{i}$ is available. In this case, one can vary the magnitude of $\Delta\pmb{u}_{i}$, i.e., the radii $\Omega_{i}$ of the uncertainty region to capture different levels of CSI imperfection. It can be noticed that by setting $\Delta\pmb{u}_{i}$ to zero, \eqref{chaCSIdef} drops down to the perfect CSI case.
	
	The corresponding SINR of the received signal at the user-$k$ and Eve under perfect and imperfect IRS to users/Eve CSI can be written, respectively, as follows:
	\begin{align} \label{SINRex}
		\gamma_i(\pmw,\btheta)&\triangleq{|\pmb{h}_{i}(\btheta)\pmw_k|^2}/{\rho_{i}}, ~ i \in \{k,e\}, \\
		\widehat{\gamma}_i(\pmw,\pmt)&\triangleq{|\widehat{\pmb{h}}_{i}(\btheta)\pmw_k|^2}/{\widehat{\rho}_{i}}, ~ i \in \{k,e\},
	\end{align}
	where  $\rho_{i}\triangleq {\sum}_{j=1,j \neq k}^{\mathcal{K}}|\pmb{h}_{i}(\btheta)\pmw_j|^2+\textcolor{black}{\sigma_i^2}$, and $\widehat{\rho}_{i}\triangleq {\sum}_{j=1,j \neq k}^{\mathcal{K}}|\widehat{\pmb{h}}_{i}(\btheta)\pmw_j|^2+\textcolor{black}{\sigma_i^2}$.
	Under the presence of Eve, the closed-form expressions for the FBR-SR of the user-$k$ under perfect and imperfect IRS to users/Eve CSI are defined as, respectively, \cite[Eq.115]{8665906}:
    \begin{align}  
		\mathcal{S}_{k}^{\mathcal{F}}(\pmw,\pmt)&\triangleq \left[C_{k}-C_{e} -\xi_k \sqrt{V_{k}} - \xi_{e} \sqrt{V_{e}}\right]^+, \label{Cstot2} \\
		\widehat{\mathcal{S}}_{k}^{\mathcal{F}}(\pmw,\pmt)&\triangleq \left[\widehat{C}_{k}-\widehat{C}_{e}-\xi_k \sqrt{\widehat{V}_{k}} - \xi_{e} \sqrt{\widehat{V}_{e}}\right]^+, \label{Cstot2ICSI} 
	\end{align}
	where $[x]^+ \triangleq \text{max}[0,x]$, while the user's data rate, and the eavesdropping rate of Eve decoding $s_k$ under perfect and imperfect CSI, are respectively given by,
	\begin{align} 
		C_{i}(\pmw,\pmt)&\triangleq\ln\left(1+\gamma_i(\pmw,\btheta)\right), i \in \{k,e\},  \\
		\widehat{C}_{i}(\pmw,\pmt)&\triangleq\ln\left(1+\widehat{\gamma}_i(\pmw,\btheta)\right), i \in \{k,e\} \label{ckedef}, 
	\end{align} 
	where $\xi_i \triangleq {Q^{-1}(\varsigma_i)}/{\text{ln}(2)\sqrt{N_t}}$, $Q^{-1}(\cdot)$ is the inverse of the Gaussian Q-function, $N_t \triangleq \mathcal{B} T$ is the packet length, $\mathcal{B}$ is the bandwidth, $T$ is the transmission duration, $\varsigma_k$ is the decoding error probability\cite{feng2021reliable}, $\varsigma_e$ is the information leakage\cite{feng2021reliable}, and $V_{k}$ and $V_{e}$ are the dispersion factors defined as \cite{feng2021reliable}: 
	\begin{align} 
		V_{i}(\pmw,\pmt) &\triangleq \frac{2 \gamma_i(\pmw,\pmt)}{(1+\gamma_i(\pmw,\pmt))} \triangleq 2\left(1-({\rho_{i}}/{\upsilon_{i}})\right), i\in \{k,e\}, \label{dispfa1} \\
		\widehat{V}_{i}(\pmw,\pmt) &\triangleq \frac{2 \widehat{\gamma}_i(\pmw,\pmt)}{(1+\widehat{\gamma}_i(\pmw,\pmt))} \triangleq 2\left(1-({\widehat{\rho}_{i}}/{\widehat{\upsilon}_{i}})\right), i\in \{k,e\}, \label{dispfa2}
	\end{align} 
	where $\upsilon_{i}\triangleq{\sum}_{j=1}^{\mathcal{K}}|\pmb{h}_{i}(\pmt)\pmw_j|^2+\textcolor{black}{\sigma_i^2}$, and $\widehat{\upsilon}_{i}\triangleq {\sum}_{j=1}^{\mathcal{K}}|\widehat{\pmb{h}}_{i}(\pmt)\pmw_j|^2+\textcolor{black}{\sigma_i^2},~i\in \{k,e\}$.  
    One can notice that ~\eqref{Cstot2ICSI} defines the secrecy rate in the finite-blocklength regime while explicitly accounting for reflected-channel uncertainty, thereby consolidating short-packet reliability/leakage penalties and channel-estimation errors into a single secrecy functional.

	Unlike the traditional SR definition in the LBR, the FBR imposes more constraints on the SR. Specifically, the reliable transmission in FBR requires that the decoding error probability $\varsigma_k$ at user-$k$ is not larger than the maximum decoding error probability in FBR  $\varsigma_{\max}$, the information leakage constraint imposes that the information leakage $\varsigma_e$ does not exceed $\varsigma_{e}^{\max}$, the maximum information leakage in FBR \cite{feng2021reliable}. Moreover, the transmission duration $T$ doesn't exceed the maximum transmission duration $T_{\max}$. Ensuring secrecy with short packets of length $N_t$ is therefore more challenging due to the compounded effects of decoding errors and
    information leakage, which limit the achievable secrecy rate. These conditions distinguish the FBR case from its LBR counterpart. One may notice that when the transmission duration $T $ approaches infinity, the dispersion factor $V_i$ reaches zero, and the FBR definition \eqref{Cstot2ICSI} reduce to the traditional LBR-SR case, which can be expressed as \cite{bloch2008wireless}: 
	\begin{align} \label{Cstot}
		\widehat{\mathcal{S}}^{\mathcal{L}}_{k}(\pmw,\btheta)& \triangleq [\widehat{C}_k(\pmw,\pmt)-\widehat{C}_e(\pmw,\pmt)]^+.  
	\end{align}
    \begin{figure*}[!t] 
		\normalsize 
		\begin{align} 
			C_e(\pmw,\thetak) &\geq q_{1,k_{e}}^{(\iota)}+n_{1,k_{e}}^{(\iota)}+2\Re \sum\limits_{j=1, j \neq k}^{\mathcal{K}}\left\{\left\la \pmb{m}_{j,k_{e}}^{(\iota)},\pmw_j\right\ra\right\} -\frac{\rho_{e}^{(\iota)}}{1+\rho_{e}^{(\iota)}}\left(\sum\limits_{j=1,j \neq k}^{\mathcal{K}}|\pmb{h}_{e}(\thetak)\pmw_j|^2\right) - \frac{1}{\upsilon_{e}^{(\iota)}}\sum\limits_{j=1}^{\mathcal{K}}|\pmb{h}_{e}(\thetak)\pmw_j|^2. \label{CEtot1} 	\tag{19} %\\
			% {C_k}(\wko,\pmt) & \geq {q}_{1,k}^{(\iota+1)}+2\Re\left\{{\sum}_{n=1}^{N}{\pmb{m}}^{(\iota+1)}_{k,k}(n) e^{j\theta_n} \right\}+ \left(e^{j \pmt}\right)^H {\pmb{\varphi}}_{1,k}^{(\iota+1)} e^{j \pmt}, \label{LBBob2x} \tag{24} \\
			% C_e(\wko,\pmt)&\geq q_{1,k_{e}}^{(\iota+1)}+n_{1,k_{e}}^{(\iota+1)}+2\mathfrak{R}\left\{{\sum}_{n=1}^{N} {\pmb{m}}^{(\iota+1)}_{k_{e}}(n) e^{(j \theta_n)} \right\}+\left(e^{j \btheta}\right)^H {\pmb{\varphi}} _{1,k_{e}}^{(\iota+1)} e^{j \btheta} + \left(e^{j \btheta}\right)^H {\pmb{\varphi}} _{k_{e},j}^{(\iota+1)} e^{j \btheta}. \label{LBEve2x} \tag{25} 
		\end{align} 
		\hrulefill 
	\end{figure*}
    
    The system model characterization highlights several fundamental challenges that arise when extending the classical LBR secrecy formulation \eqref{Cstot} to the finite blocklength regime \eqref{Cstot2ICSI}. First, the Shannon capacity based secrecy model used in the LBR case becomes invalid in FBR, since short packet transmission introduces non-negligible decoding error $\varsigma_k$ and information leakage $\varsigma_e$ penalties. Second, secrecy, reliability, and latency become tightly coupled through the dispersion terms in \eqref{Cstot2ICSI}, making these performance metrics interdependent rather than separable as in the LBR counterpart. Third, the square-root dispersion factors in \eqref{Cstot2ICSI} create non-separable and concave components in the secrecy expression, eliminating the monotonicity and convexity properties relied upon by conventional LBR optimization approaches. Finally, under imperfect CSI, both the achievable rate and the dispersion depend on the magnitude of the CSI uncertainty, leading to semi-infinite constraints and a significantly smaller feasible region. These challenges fundamentally distinguish FBR-based secrecy optimization from the traditional LBR framework and motivate the algorithmic developments presented in the subsequent sections.
   
	\section{SSR Maximization under perfect CSI in FBR systems} \label{sec3}
	In this section, we first address the problem of maximizing the SSR under the perfect CSI assumption, which is achieved by setting $\Delta\pmb{u}_{i}$ to zero in \eqref{ch2-h1}. The optimization problem can be formally stated as: 
	\begin{subequations} \label{P1}
		\begin{alignat}{3} 
			&(\mathcal{P}1):&& ~\underset{\pmw,\btheta}{\max}  ~ R_{s}(\pmw,\btheta), \label{P1-1} \\            
			&~\text{s.t.} && ~{\sum}_{k = 1}^{\mathcal{K}}\Vert \pmw_k \Vert^2 \leq P_T,\label{P1-2} \\
			& && |e^{(j \btheta)}| = 1,  \label{P1-3} \\
			& && \varsigma_k\leq \varsigma_{\max}, ~\varsigma_e \leq \varsigma_{e}^{\max}, ~T \leq T_{\max},  \label{P1-4}
		\end{alignat}
	\end{subequations}
	where $R_{s}(\pmw,\btheta)={\sum}_{k=1}^{\mathcal{K}} \mathcal{S}^{\mathcal{F}}_{k}(\pmw,\btheta)$ is the SSR, \textcolor{black}{$P_T$} is Alice's power budget, $\eqref{P1-2}$ captures the sum of the transmitted power constraint, and $\eqref{P1-3}$ captures the UMC of the PREs' phase shift. 
    It should be noted that the FBR constraints in \eqref{P1-4} are inherently embedded within the objective function, as the SR expression $R_s(\pmw,\pmt)$ incorporates the square-root of the dispersion factors, as shown in~\eqref{Cstot2ICSI}. Consequently, the FBR-SSR optimization problem falls outside the convex domain, making it more difficult to solve than its LBR counterpart.
    
	The optimization problem ($\mathcal{P}1$) is non-convex since the objective function \eqref{P1-1} is not concave and the UMC (\ref{P1-3}) is non-convex. To tackle this non-convex problem, one can employ the AO technique \cite{nguyen2022leveraging}. Specifically, at iteration $(\iota)$, the feasible point $(\wk,\thetak)$ is generated from ($\mathcal{P}1$) by solving two sub-problems, the first one is to optimize $\pmw$ with a fixed $\btheta$:
	\begin{align}  \label{probw1}
		(\mathcal{P}1.1): ~ \underset{\pmw}{\max} ~ R_{s}(\wk,\btheta), ~ \text{s.t.} ~ (\ref{P1-2}),
	\end{align} 
	then, we optimize $\btheta$ with a fixed $\pmw$ by solving the following problem
	\begin{align} \label{probth}
		(\mathcal{P}1.2): ~ \underset{\btheta}{\max} ~ R_{s}(\pmw,\thetak), ~ \text{s.t.} ~ (\ref{P1-3}).
	\end{align}
	
	However, with AO, solving two subproblems $(\mathcal{P}1.1)$ and $(\mathcal{P}1.2)$ is computationally demanding, especially given the large number of PREs of the IRS.  Our proposed linearization method in the sequel uses mathematically tractable approximation functions, leading to a computationally efficient algorithm that can be used for a large number of IRSs' PREs.
	
	\subsubsection{Sub-Problem for Optimizing the Beamforming Vectors} \label{subsec1}
	We fix $\btheta$ given $\wk$ and solve the problem ($\mathcal{P}1$) to obtain $\wko$ satisfying $R_s(\wko,\thetak) \geq R_s(\wk,\thetak)$. 
	We start by linearizing the objective function in (\ref{P1-1}), which consists of four parts: the user-$k$'s SR $C_k(\pmw,\pmt)$, Eve's negative eavesdropping rate $C_e(\pmw,\pmt)$, the user-$k$ dispersion factor $V_k(\pmw,\pmt)$, and Eve's dispersion factor $V_e(\pmw,\pmt)$ as shown in \eqref{Cstot2ICSI}. 
	
	First, we adopt the idea in \cite{TTN16} to convert user-$k$'s SR into a linear form. Specifically, using the inequality (\ref{fund1}) in Appendix \ref{AppA}, let's define $\pmb{\Lambda}\triangleq \pmb{h}_{k}(\btheta^{(\iota)})\pmw_k$, $\pmb{\digamma}\triangleq \rho_k$, $\hat{\Lambda}\triangleq \pmb{h}_{k}(\btheta^{(\iota)})\pmw_k^{(\iota)}$ and
	$\hat{\digamma} \triangleq \rho_k$, hence, the users' data rate can be written as: 
	\begin{equation} \label{LBBob} 
		C_k(\pmw,\thetak)\geq q_{1,k}^{(\iota)}+2\Re\left\{\left\la \pmb{m}_{k,k}^{(\iota)},\pmw_k\right\ra\right\}-n_{1,k}^{(\iota)}\upsilon_{_{k}}^{(\iota)},
	\end{equation}
	where
	\begin{alignat*}{2}
		&q_{1,k}^{(\iota)} &&\triangleq C_k(\wk,\thetak)-\gamma_k(\wk,\thetak)-\sigk n_{1,k}^{(\iota)}, \\
		&\pmb{m}^{(\iota)}_{k,k} &&\triangleq {\pmb{h}_{_{k}}^H (\thetak)\pmb{h}_{_{k}}(\thetak)\wk_{k}}/{\rho_{k}^{(\iota)}}, \\
		&n_{1,k}^{(\iota)} &&\triangleq{1}/{\rho_{k}^{(\iota)}}-{1}/{\upsilon_{_{k}}^{(\iota)}}.
	\end{alignat*}
	
	Second, we linearize Eve's eavesdropping rate, which can be written as: 
	\begin{align}\label{EveSRt1}
		-\ln(1+\gamma_e(\pmw,\thetak))&\triangleq \overset{a_1}{\overbrace{\ln(1+\rho_{e}^{(\iota)})}}- \overset{a_2}{\overbrace{\ln(1+\upsilon_e^{(\iota)})}}.  
	\end{align}
	Here, we adopt the idea in \cite{niu2022joint}, where the term ($a_1$) in (\ref{EveSRt1}) can be linearized by defining $\pmb{z}\triangleq \rho_{e}$ and substituting it in  the inequality (\ref{fund2}) in Appendix \ref{AppA}. The term ($a_2$) in (\ref{EveSRt1}) can be linearized by defining $\pmb{\Upsilon}\triangleq \upsilon_e$ and substituting it in the inequality (\ref{fund3}) in Appendix \ref{AppA}. Hence, Eve's eavesdropping rate can be expressed as in \eqref{CEtot1}, where 
	\begin{align*}
		q_{1,k_{e}}^{(\iota)} &\triangleq \ln\left(\rho_{e}^{(\iota)}\right)-\rho_{e}^{(\iota)}-\ln\left(\upsilon_{e}^{(\iota)}\right)+1, \\
		n_{1,k_{e}}^{(\iota)} &\triangleq\left(-\frac{\rho_{e}^{(\iota)}}{1+\rho_{e}^{(\iota)}}-\frac{1}{\upsilon_{e}^{(\iota)}}\right)\sige, \\
		\pmb{m}_{k_{e},j}^{(\iota)} &\triangleq \pmb{h}_{e}^H(\thetak)\pmb{h}_{e}(\thetak) \pmw_j^{(\iota)}.
	\end{align*}
	% Next, we linearize the FBR dispersion factors $V_{i}(\pmw,\pmt), i \in \{k,e\}$. Using \eqref{dispfa2}, the inequalities \eqref{xy} and \eqref{xt} in Appendix \ref{AppA}, and by defining $x\triangleq V_i(\wk,\pmt)$, $\mathbf{A} \triangleq \rho_{i}$, and $B \triangleq \varsigma_i$, the dispersion factors can be expressed as:  \setcounter{equation}{19}
    Next, we linearize the FBR dispersion factors $V_{i}(\pmw,\pmt), i \in \{k,e\}$. Using \eqref{dispfa1}, the dispersion factors $V_i(\pmw,\pmt)$ can be expressed as functions of the variables $(\rho_i,\upsilon_i)$. Then, by applying the inequalities \eqref{xy} and \eqref{xt} in Appendix \ref{AppA}, the dispersion factors can be expressed as: \setcounter{equation}{19}
	\begin{align}  \label{urllcbob}
		\xi_{i} \sqrt{V_{i}(\pmw,\thetak)} \leq   q_{2,i}^{(\iota)}-2 \sum\limits_{j=1, j \neq k}^{\mathcal{K}}\Re\left\{\left\la \pmb{d}_{j,i}^{(\iota)},\pmw_j\right\ra\right\} -n_{2,i}^{(\iota)}\upsilonk, 
	\end{align}
     \begin{figure*}[!t] 
		\normalsize 
		\begin{align} 
			% C_e(\pmw,\thetak) &\geq q_{1,k_{e}}^{(\iota)}+n_{1,k_{e}}^{(\iota)}+2\Re \sum\limits_{j=1, j \neq k}^{\mathcal{K}}\left\{\left\la \pmb{m}_{j,k_{e}}^{(\iota)},\pmw_j\right\ra\right\} -\frac{\rho_{e}^{(\iota)}}{1+\rho_{e}^{(\iota)}}\left(\sum\limits_{j=1,j \neq k}^{\mathcal{K}}|\pmb{h}_{e}(\thetak)\pmw_j|^2\right) - \frac{1}{\upsilon_{e}^{(\iota)}}\sum\limits_{j=1}^{\mathcal{K}}|\pmb{h}_{e}(\thetak)\pmw_j|^2, \label{CEtot1} 	\tag{15} \\
			{C_k}(\wko,\pmt) & \geq {q}_{1,k}^{(\iota+1)}+2\Re\left\{{\sum}_{n=1}^{N}{\pmb{m}}^{(\iota+1)}_{k,k}(n) e^{j\theta_n} \right\}+ \left(e^{j \pmt}\right)^H {\pmb{\varphi}}_{1,k}^{(\iota+1)} e^{j \pmt}, \label{LBBob2x} \tag{23} \\
			C_e(\wko,\pmt)&\textcolor{black}{\leq} q_{1,k_{e}}^{(\iota+1)}+n_{1,k_{e}}^{(\iota+1)}+2\mathfrak{R}\left\{{\sum}_{n=1}^{N} {\pmb{m}}^{(\iota+1)}_{k_{e}}(n) e^{(j \theta_n)} \right\}+\left(e^{j \btheta}\right)^H {\pmb{\varphi}} _{1,k_{e}}^{(\iota+1)} e^{j \btheta} + \left(e^{j \btheta}\right)^H {\pmb{\varphi}} _{k_{e},j}^{(\iota+1)} e^{j \btheta}. \label{LBEve2x} \tag{24} \\
            \xi_{i} \sqrt{V_{i}(\wko,\pmt)} &\leq   q_{2,i}^{(\iota+1)}-2\Re\left\{ {\sum}_{n=1}^N \pmb{d}_{i}^{(\iota+1)} e^{(j \theta_n)}\right\} +\left(e^{(j \pmt)}\right)^H \pmb{\varphi}_{2,i}^{(\iota+1)} e^{(j \pmt)}, i \in \{k,e\}, 
			\label{urllcbob2} \tag{25}\\
			q_{2,i}^{(\iota+1)}&\triangleq\frac{\xi_i{\left(\left(\upsilon^{(\iota+1)}_i\right)^2 \sqrt{V_i(\wko,\thetak)}+2\left(\upsilon^{(\iota+1)}_i\right)^2+2\sigma_{i}\rho_{i}^{(\iota+1)}-4\textcolor{black}{\sigma_i^2} \left(\upsilon^{(\iota+1)}_i\right)\right)}}{{\left(2\left(\upsilon^{(\iota+1)}_i\right)^2 \sqrt{V_i(\wko,\thetak)}\right)}}, \label{q2k} \tag{26}\\
			\pmb{d}_{i,j}^{(\iota+1)}(n) &\triangleq  {\xi_i}\bigg/{\left(\upsilon^{(\iota+1)}_i \sqrt{V_i(\wko,\thetak)}\right)}  \left(\wko_j \right)^H \pmb{h}_{i}^H(\thetak) \pmb{u}_i \Xi_n {\pmb{G}}_{_{\text{{AR}}}} \wko_j ,n \in N, j\in K, j \neq k,\label{mkj1} \tag{27}\\
			q_{k}^{(\iota+1)} &\triangleq {q5}_{1,k}^{(\iota+1)}+{q}_{1,k_{e}}^{(\iota+1)}+ q_{2,k}^{(\iota+1)}+ q_{2,k_{e}}^{(\iota+1)}+n_{1,ke}^{(\iota+1)}- \left(e^{j \thetak}\right)^H \pmb{\varphi} _{k}^{(\iota+1)}e^{j\btheta}-2\lambda_{\text{max}}\left(\pmb{\varphi}_{k}^{(\iota+1)}\right) N, \tag{29} \label{qk2}\\
			{\pmb{m}}^{(\iota+1)}_k(n) &\triangleq {\pmb{m}}^{(\iota+1)}_{k,k}(n)+{\pmb{m}}^{(\iota+1)}_{k_{e}}(n)+{\pmb{d}}^{(\iota+1)}_{k}(n)+{\pmb{d}}^{(\iota+1)}_{k_{e}}(n)+ {\sum}_{m=1}^{N}e^{-j \theta_m^{(\iota)}} \pmb{\varphi}_{k}^{(\iota+1)}(m,n)+\lambda_{\text{max}}\left(\pmb{\varphi}_{k}^{(\iota+1)}\right).\tag{30} \label{mk2}
		\end{align} 
		\hrulefill 
	\end{figure*}
	where
	\begin{align*}
		q_{2,i}^{(\iota)}&\triangleq \xi_{i}\left(\frac{\sqrt{V_i(\wk,\thetak)}}{2}+\frac{\left(\upsilon^{(\iota)}_i\right)^2+\rho^{(\iota)}_i \sigma_{i}-2 \upsilon^{(\iota)}_i \sigma_{i}}{\left(\upsilon^{(\iota)}_i\right)^2 \sqrt{V_i(\wk,\thetak)}}\right), \nonumber \\
		\pmb{d}_{j,i}^{(\iota)} &\triangleq  \frac{\xi_i [\pmb{h}_{i}(\thetak)]^2\wk_j}{\upsilon^{(\iota)}_i \sqrt{V_i(\wk,\thetak)}}, i\in \{k,e\}, \nonumber \\
		n_{2,i}^{(\iota)} & \triangleq \frac{\xi_{i} \rho^{(\iota)}_i}{\left(\upsilon^{(\iota)}_i\right)^2 \sqrt{V_i(\wk,\thetak)}}, i\in \{k,e\}. \nonumber
	\end{align*}
	
	By substitution \eqref{LBBob}, \eqref{CEtot1}, and \eqref{urllcbob} into \eqref{Cstot2ICSI},  the approximated surrogate FBR-SR function can be written as:  
	\begin{align} \label{SRFBRtot}
		\widetilde{\mathcal{S}}_{k}^{\mathcal{F}}(\pmw,\thetak)&\geq {q}_{k}^{(\iota)}+2\Re\left\{\left\la {\pmb{m}}_{k}^{(\iota)},\pmw_{k}\right\ra\right\}-(\pmw_{k})^H {\pmb{\psi}}^{(\iota)}_{k} \pmw_{k},
	\end{align}
	where \begin{align*}
		{q}_{k}^{(\iota)} &\triangleq q_{1,k}^{(\iota)}+q_{1,k_{e}}^{(\iota)}+n_{1,k_{e}}^{(\iota)}+q_{2,k}^{(\iota)}+q_{2,k_{e}}^{(\iota)}, \\
		{\pmb{m}}_{k}^{(\iota)} &\triangleq {\sum}_{j=1}^{\mathcal{K}}\pmb{m}_{j,k}^{(\iota)}+{\sum}_{j=1,j \neq k}^{\mathcal{K}} \left(\pmb{m}_{k_{e},j}^{(\iota)}-\pmb{d}_{j,k}^{(\iota)}-\pmb{d}_{j,k_{e}}^{(\iota)}\right), \\
		\pmb{\psi}^{(\iota)}_{k}
		& \triangleq {\sum}_{j=1}^{\mathcal{K}} n_{j}^{(\iota)} \pmb{h}_{j}^H(\thetak) \pmb{h}_{j}(\thetak) + n_{k_{e}}^{(\iota)} \pmb{h}_{e}^H(\thetak) \pmb{h}_{e}(\thetak), \\
		n_{k}^{(\iota)} &\triangleq n_{1,k}^{(\iota)}+n_{2,k}^{(\iota)}, \\
		n_{k_{e}}^{(\iota)} &\triangleq \frac{\rho_{e}^{(\iota)}}{\left(1+\rho_{e}^{(\iota)}\right)}+\frac{1}{\upsilon_{e}^{(\iota)}}+n_{2,k_{e}}.
	\end{align*}

	Finally, by substitution  \eqref{SRFBRtot} in \eqref{probw1}, we can recast problem ($\mathcal{P}1.1$) as:
	\begin{subequations} \label{solw}
		\begin{alignat}{3} 
			&(\mathcal{P}1.3): && \min_{\pmw} ~ -{\sum}_{k=1}^{\mathcal{K}}\widetilde{\mathcal{S}}_{k}^{\mathcal{F}}(\pmw,\thetak), \label{solw-1} \\            
			&~\text{s.t.} &&  \eqref{P1-2}, ~\eqref{P1-4}.  \label{solw-12x}
		\end{alignat}
	\end{subequations}
   
	Problem ($\mathcal{P}1.3$) is a convex problem that can be efficiently solved using standard solvers, e.g., the interior point method or the CVX toolbox \cite{grant2014cvx} to generate $\wko$.  
	\subsubsection{Sub-Problem for Optimizing the PREs} \label{subsec2}
	Likewise, given $\pmw$, we aim to find $\btheta^{(\iota+1)}$ such that, $\mathcal{S}_{k}(\wko_k,\thetako) \geq \mathcal{S}_{k}(\wko_k,\thetak).$ 
	Similar to the previous section, the lower bound approximation of the user's SR can be obtained using the inequality \eqref{fund1} in Appendix \ref{AppA}. The user's data rate can be expressed as in \eqref{LBBob2x}, where, 
	\begin{alignat*}{2}
		&{q}_{1,k}^{(\iota+1)} &&\triangleq C_k(\wko,\thetak)-\gamma_k(\wko,\thetak)-\sigk {n}_{1,k}^{(\iota+1)},\\
		&{n}_{1,k}^{(\iota+1)} &&\triangleq{1}/{\rho_{k}^{(\iota+1)}}-{1}/{\upsilon_{_{k}}^{(\iota+1)}},\\
		&{\pmb{m}}^{(\iota+1)}_{k,k}(n) &&\triangleq {\hat{\pmb{m}}^{(\iota+1)}_{k,k}(n)}/{\rho_{k}^{(\iota+1)}}, \\
		&\hat{\pmb{m}}^{(\iota+1)}_{k,k}(n) &&\triangleq \left(\wko_k \right)^H\pmb{h}_{_{k}}^H(\thetak) {\pmb{u}}_k  \pmb{\Xi}_n \pmb{G}_{_{\text{AR}}}  \wko_k, \\
		&\pmb{\varphi}_{1,k}^{(\iota+1)} &&\triangleq -{n}_{1,k}^{(\iota+1)}{\sum}_{j=1}^{\mathcal{K}} \pmb{\varphi}_{k,j}^{(\iota+1)}, \\
		&\pmb{\varphi}_{k,j}^{(\iota+1)} &&\triangleq \left(\pmb{\aleph}_{k,j}^{(\iota+1)}(n)\right)^* \pmb{\aleph}_{k,j}^{(\iota+1)}(m), n\in N, m \in N, \\
		&\pmb{\aleph}_{k,j}^{(\iota+1)}(n) &&\triangleq \pmb{u}_k \pmb{\Xi}_n \pmb{G}_{_{\text{AR}}} \wko_j.
	\end{alignat*}
	Similarly, we can express Eve's eavesdropping rate as in \eqref{LBEve2x}, where, 
	\begin{alignat*}{2}
		&q_{1,k_{e}}^{(\iota+1)} &&\triangleq \ln\left(\rho_{e}^{(\iota+1)}\right)-\rho_{e}^{(\iota+1)}-\ln\left(\upsilon_{e}^{(\iota+1)}\right)+1, \\
		&n_{1,k_{e}}^{(\iota+1)} &&\triangleq\left(-\frac{\rho_{e}^{(\iota+1)}}{1+\rho_{e}^{(\iota+1)}}-\frac{1}{\upsilon_{e}^{(\iota+1)}}\right)\sige, \\
		&{\pmb{m}}^{(\iota+1)}_{k_e}(n) &&\triangleq {\sum}_{j=1,j \neq k}^{\mathcal{K}} {\pmb{m}}^{(\iota+1)}_{j,k_{e}}(n), \\
		&{\pmb{m}}^{(\iota+1)}_{j,k_{e}}(n) &&\triangleq \left(\wko_k \right)^H\pmb{h}_{e}^H(\thetak) \pmb{u}_{e}  \pmb{\Xi}_n {\pmb{G}}_{_{\text{{AR}}}}  \wko_k, \\
		&\pmb{\varphi} _{1,k_{e}}^{(\iota+1)} &&\triangleq\left(\frac{-\rho_{e}^{(\iota+1)}}{(1+\rho_{e}^{(\iota+1)})}-\frac{1}{(\sige+\upsilon_{e}^{(\iota+1)})}\right)\left(\sum\limits_{j=1}^{\mathcal{K}} {\pmb{\varphi}}_{k_{e},j}^{(\iota+1)} \right), \\
		&\pmb{\varphi}_{k_{e},j}^{(\iota+1)} &&\triangleq \left({\pmb{\aleph}}_{k_{e},j}^{(\iota+1)}(n)\right)^* \pmb{\aleph}_{k_{e},j}^{(\iota+1)}(m), \\
		&\pmb{\aleph}_{k_{e},j}^{(\iota+1)}(n) &&\triangleq \pmb{u}_e \pmb{\Xi}_n {\pmb{G}}_{_{\text{{AR}}}} \wko_j.
	\end{alignat*}

	Next, we express the user-$k$'s and Eve's dispersion factor. For $i \in \{k,e\}$, we can express the dispersion factors as in \eqref{urllcbob2},
	where $q_{i,k}^{(\iota+1)}$ is expressed as in \eqref{q2k}, $\pmb{d}_{i,j}^{(\iota+1)}(n)$ is expressed as in \eqref{mkj1}, and 
	\begin{align*}
		\pmb{d}_{i}^{(\iota+1)} &\triangleq{\sum}_{j=1, j \neq k}^\mathcal{K} \pmb{d}_{i,j}^{(\iota+1)}, \\
		\pmb{\varphi}_{2,i}^{(\iota+1)}&\triangleq{n}_{2,i}^{(\iota+1)} {\sum}_{j=1}^{\mathcal{K}} \pmb{\varphi}_{i,j}^{(\iota+1)}, \\
		{n}_{2,i}^{(\iota+1)}&\triangleq {(\xi_{i} \rho^{(\iota+1)}_i)}/{\left(\left(\upsilon^{(\iota+1)}_i\right)^2 \sqrt{V_i(\wko,\thetak)}\right)}.
	\end{align*}
    \begin{algorithm}
		\caption{Proposed AO Algorithm for Solving Problem $(\mathcal{P}1)$} \label{alg1}
		\label{alg:iterative}
		\begin{algorithmic}[1]
			\State \textbf{Input:} Initial values $(\pmb{w}^{(1)},~\btheta^{(1)}$, ~$\varsigma_{\max}$,~$\varsigma_{e}^{\max}$,~$T_{\max}$), convergence tolerance $\epsilon_t > 0$.
			\State \textbf{Initialize:} Set iteration index $\iota = 1$.
			\Repeat
			\State Update $\pmb{w}^{(\iota+1)}$ using \eqref{solw}, and $\btheta^{(\iota+1)}$ using \eqref{solth}.
			\If{$\dfrac{\left| R_s(\pmb{w}^{(\iota+1)}, \btheta^{(\iota+1)}) - R_s(\pmb{w}^{(\iota)}, \btheta^{(\iota)}) \right|}{R_s(\pmb{w}^{(\iota)}, \btheta^{(\iota)})} \leq \epsilon_t$}
			\State Set $\pmw^* \gets \pmb{w}^{(\iota+1)}$, $\pmt^* \gets \btheta^{(\iota+1)}$;
			\State \textbf{Output:} $(\pmw^*, \pmt^*)$; \textbf{terminate}.
			\Else
			\State $\iota \gets \iota + 1$
			\EndIf
			\Until{convergence}
		\end{algorithmic}
	\end{algorithm}
	By substituting \eqref{LBBob2x}, \eqref{LBEve2x}, and \eqref{urllcbob2} in \eqref{Cstot2ICSI}, the lower bounding concave approximation of the FBR-SR can be expressed  as:
	\setcounter{equation}{27}
	\begin{align} \label{SRtheta}
		\mathcal{S}_k^{\mathcal{F}}\left(\wko,\btheta\right)& \geq q_{k}^{(\iota+1)}+ 2\sum\limits_{n=1}^{N}\Re\left\{\pmb{m}^{(\iota+1)}_{k}(n) e^{j \theta_n}\right\},
	\end{align}
	where $	q_{k}^{(\iota+1)}$ is expressed as in \eqref{qk2}, ${\pmb{m}}^{(\iota+1)}_k(n)$ is expressed as in \eqref{mk2}, and \setcounter{equation}{30} 
	\begin{align}
		\pmb{\varphi}_{k}^{(\iota+1)} &\triangleq \pmb{\varphi}_{1,k}^{(\iota+1)}+\pmb{\varphi}_{1,k_{e}}^{(\iota+1)}+\pmb{\varphi}_{k_{e},j}^{(\iota+1)}+\pmb{\varphi}_{2,k}^{(\iota+1)}+\pmb{\varphi}_{2,k_{e}}^{(\iota+1)}. 
	\end{align}
	With \eqref{SRtheta}, we formulate the following optimization problem to obtain $\thetako$,
    \setcounter{equation}{31}
	\begin{subequations} \label{Ptheta}
		\begin{alignat}{2} 
			&(\mathcal{P}1.4): && ~ \underset{\btheta}{\max}~ \widehat{R}_{s}\left(\wko,\btheta\right),\label{Ptheta-1} \\
			&~\text{s.t.} &&  \eqref{P1-3},~\eqref{P1-4}.  
		\end{alignat}
	\end{subequations}
	To tackle problem ($\mathcal{P}1.4$) we define,  
	\begin{align} \label{convx2}
		\theta_n^{(\iota+1),k} \triangleq 2\pi - \angle {\pmb{m}}^{(\iota+1)}_k(n), n=1,\dots,N,
	\end{align}
	then we find $\btheta^{(\iota+1)}$ that maximize the $\mbox{SR}(\wko,\pmt)$ and satisfies the UMC constraint \eqref{P1-3},
	\begin{align} \label{solth}
		\btheta^{(\iota+1)} =  2\pi - \angle \pmb{m}^{(\iota+1)}(n), n=1,\dots,N.
	\end{align}
	One can notice that by setting $V_i = 0$, the above FBR optimization problem $(\mathcal{P}1)$ is reduced to the LBR optimization problem, which can be solved similarly.
	The procedure to solve problem ($\mathcal{P}1$) is described in Algorithm \ref{alg1}, which converges to a locally optimal solution of ($\mathcal{P}1$) as formally stated in the following theorem.
	\begin{theorem} \label{theo1}
		The obtained solution by Algorithm \ref{alg1} is a locally optimal solution for problem $(\mathcal{P}1)$.
	\end{theorem}
	\textit{Proof}: See Appendix \ref{AppB}.
	\subsubsection{Complexity Analysis}
	The developed Algorithm \ref{alg1} is designed to tackle problem $(\mathcal{P}1)$ by decoupling the beamforming vector and the IRS's PREs within the objective function. The resulting problem $(\mathcal{P}1)$ is solved by decomposing it into two sub-problems. The beamforming sub-problem $(\mathcal{P}1.3)$ is convex and can be efficiently solved using standard solvers such as the interior point method or the CVX toolbox \cite{grant2014cvx}, while the IRS phase shift sub-problem $(\mathcal{P}1.4)$ admits a closed-form solution. The algorithm's complexity can be estimated by its worst-case runtime and the number of decision variables \cite{labit2002sedumi}. These complexity expressions are based on the size of the optimization variables involved in each sub-problem. Thus, in Algorithm \ref{alg1}, the computational complexity of obtaining $\pmw$ given $\btheta$ is $\mathcal{O} (M^3)$, and the computational complexity of obtaining $\btheta$ given $\pmw$ is \textcolor{black}{$\mathcal{O}(N^{3})$}. 
    
	\section{SSR Maximization under Imperfect CSI with FBR systems} \label{sec4}
	To deal with the imperfect CSI from the IRS to the users and Eve, we adopt the channel modeling in equation \eqref{ICSIcha} where only partial/imperfect CSI $\widehat{\pmb{u}}_{i}$ is available. First, we can cast the SSR maximization problem as follows: %\setcounter{equation}{34}
	\begin{subequations} \label{P6}
		\begin{alignat} {2}
			&(\mathcal{P}2): && \underset{\pmw,\btheta}{\max}  ~ \widehat{R}_{s}(\pmw,\btheta), \label{P2-1} \\ 
			&~\text{s.t.} &&  {\sum}_{k = 1}^{\mathcal{K}}\Vert \pmw_k \Vert^2 \leq P_T, \label{P2-2} \\ 
			& && |e^{(j \btheta)}| = 1, \label{P2-3} \\
			& && \|\Delta\pmb{u}_i\|_2 \leq \xi_{i}, ~i \in \{k,e\}, \label{P2-4}\\
			& && \varsigma_k\leq \varsigma_{\max}, ~\varsigma_e \leq \varsigma_{e}^{\max}, ~T \leq T_{\max}, \label{P2-5}
		\end{alignat}
	\end{subequations}
	where $ \widehat{R}_s(\pmw,\pmt) \triangleq {\sum}_{k}^{\mathcal{K}}\widehat{\mathcal{S}}_{k}^{\mathcal{F}}(\pmw,\pmt)$ is the SSR, $P_T$ is Alice's power budget, \eqref{P2-3} captures the sum of the transmitted power constraint, \eqref{P2-4} captures the UMC of the PREs' phase shift, and \eqref{P2-5} captures the FBR transmission duration and latency constraints. Unlike the traditional LBR case under imperfect CSI, problem $(\mathcal{P}2)$ is more challenging. First, this is because the additional FBR constraints that are captured by the non-convex dispersion factor expression \eqref{dispfa2} are embedded within the objective function $\widehat{R}_s(\pmw,\pmt)$. Moreover, due to the imperfect CSI, as modeled in \eqref{ICSIcha}, the semi-infinite constraint \eqref{P2-4} further complicates the problem.
    It should be noted that the FBR constraints in \eqref{P2-5} are inherently embedded within the objective function, as the SR expression $\widehat{R}_s(\pmw,\pmt)$ incorporates the square-root of the dispersion factors, as shown in~\eqref{Cstot2ICSI}. Consequently, the FBR-SSR optimization problem falls outside the convex domain, making it more difficult to solve than its LBR counterpart.
    
	The optimization problem ($\mathcal{P}2$) is non-convex since the objective function \eqref{P2-1} is not concave and the UMC \eqref{P2-3} is non-convex. To tackle this non-convex problem, one can employ the AO technique \cite{nguyen2022leveraging}. Specifically, at iteration $(\iota)$, the feasible point $(\wk,\thetak)$ is generated from ($\mathcal{P}2$) by solving two sub-problems, the first one is to optimize $\pmw$ with a fixed $\btheta$:
	\begin{align}  \label{probw}
		(\mathcal{P}2.1): ~ \underset{\pmw}{\max} ~ \widehat{R}_{s}(\wk,\btheta), ~ \text{s.t.} ~ \eqref{P2-2},~\eqref{P2-4},~\eqref{P2-5},
	\end{align} 
	then, we optimize $\btheta$ with a fixed $\pmw$ by solving the following problem
	\begin{align} \label{probth}
		(\mathcal{P}2.2): ~ \underset{\btheta}{\max} ~ \widehat{R}_{s}(\pmw,\thetak), ~ \text{s.t.} ~ \eqref{P2-3},~\eqref{P2-4},~\eqref{P2-5}.
	\end{align}
	\begin{figure*}[!t] 
		\normalsize 
		\begin{align}
			&\mathbf{X}_k  \triangleq \mathbf{\Phi} \pmb{G}_{_\text{AR}} \pmw_k \pmw_k^{(\iota),H} \pmb{G}_{_\text{AR}}^H \mathbf{\Phi}^{(\iota),H}+ \mathbf{\Phi}^{(\iota)} \pmb{G}_{_\text{AR}} \pmw_k^{(\iota)} \pmw_k^{H} \pmb{G}_{_\text{AR}}^H \mathbf{\Phi}^{(\iota),H} - \mathbf{\Phi}^{(\iota)} \pmb{G}_{_\text{AR}} \pmw_k^{(\iota)} \pmw_k^{(\iota),H} \pmb{G}_{_\text{AR}}^H \mathbf{\Phi}^{(\iota),H} \label{LMI1-1}, \tag{42} \\
			&\Delta\pmb{u}_{k} \mathbf{X}_k \Delta\pmb{u}_{k}^H + 2 \Re \{( \widehat{\pmb{u}}_{k}^H \mathbf{X}_{k}) \Delta\pmb{u}_{k}\}+ d_k \geq (2^{\varphi_k}-1)\beta_{k}, \forall \|\Delta\pmb{u}_k\|_2 \leq \xi_k, \forall k. \tag{43} \label{Bobx1} 
		\end{align}
		\hrulefill 
	\end{figure*}
	However, with AO, solving two subproblems $(\mathcal{P}2.1)$ and $(\mathcal{P}2.2)$ is computationally demanding, especially given the large number of PREs of the IRS.  Our proposed linearization method in the sequel uses mathematically tractable approximation functions, leading to a computationally efficient algorithm that can be used for a large number of IRSs' PREs.

    The feasibility of the proposed sub-problems $(\mathcal{P}2.1)$ and $(\mathcal{P}2.2)$ under the AO framework is affected by the incorporation of the FBR constraints, in contrast to the conventional LBR case. In LBR-based formulations, the achievable SR depends solely on the channel capacities of the legitimate and eavesdropping links as it can be seen in \eqref{Cstot}, resulting in a relatively larger feasible region and a convex or quasi-convex optimization landscape. However, in the FBR regime, the users' SR is penalized by the dispersion factors $\widehat{V}_i(\pmw,\pmt)$ and the decoding error and leakage probabilities $\varsigma_k$ and $\varsigma_e$, which are inherently coupled with the system latency and transmission duration through the term $\xi_i \triangleq {Q^{-1}(\varsigma_i)}/{(\text{ln}(2)\sqrt{N_t})}$, as it can be noticed in \eqref{Cstot2ICSI}. These additional terms introduce non-convexity and significantly tighten the feasible region.
    
    Specifically, in the beamforming sub-problem $(\mathcal{P}2.1)$, the FBR constraints embedded in the objective and constraint sets (\ref{P2-4})-(\ref{P2-5}) reduce the flexibility of $\pmw$ compared to the LBR case, where the optimization only considers the transmit power and SINR thresholds. Likewise, in the IRS PREs sub-problem $(\mathcal{P}2.2)$, the feasibility of $\btheta$ becomes more constrained as the dispersion-related terms $\widehat{V}_i(\pmw,\pmt)$ create additional coupling between $\btheta$ and $\pmw$. As a result, while the LBR sub-problems generally guarantee feasibility for a broad range of initialization points, the FBR sub-problems exhibit a smaller feasible region and a more complex curvature of the objective space.

    Nevertheless, through the SCA, the $S$-procedure, and the PCCP framework, we ensure that each iteration operates on a feasible convex surrogate, and the final solution satisfies the feasible set defined by the FBR constraints. This guarantees that the proposed AO-based algorithm maintains feasibility across iterations and converges to a locally optimal solution that satisfies both secrecy and FBR requirements.
    
	To solve problem $(\mathcal{P}2)$, we first introduce slack variables to decompose the coupling of the beamforming vector and the IRS's PREs in the objective function to facilitate the AO method. Specifically, we first substitute \eqref{Cstot} into \eqref{P2-1}. Then we introduce slack variables $z$ as the FBR-SSR's lower bound, ${\varphi}_k$ represents the minimum users' rate, $\mu_{k_e}$ represents the maximum eavesdropping rate of Eve, and $ \tilde{\varphi}_i, i \in \{k,e\}$ represents the maximum dispersion factor for the user-$k$ and Eve, respectively. The problem $(\mathcal{P}2)$ can be recast as: 
	\begin{subequations} \label{P6.1}
		\begin{alignat}{7} 
			&(\mathcal{P}2.3): && ~ \underset{\pmw,\btheta}{\max} ~ z, \label{P6-1} \\ 
%			&~\text{s.t.} &&  z \leq \widehat{R}_s(\pmw,\pmt), \label{PwCSI-2x} \\ 
			&~\text{s.t.} && z \leq {\sum}_{k}^{\mathcal{K}}\varphi_k-\mu_{k_e}-\tilde{\varphi}_k-\tilde{\varphi}_{k_e}, \forall k, \label{PwCSI-2x2} \\ 
			& && {\varphi}_k \leq \widehat{C}_k, ~\forall \|\Delta\pmb{u}_k\|_2 \leq \Omega_{k}, ~\forall k, \label{P6-2} \\
			& && \mu_{k_e} \geq \widehat{C}_e, ~\forall \|\Delta\pmb{u}_e\|_2 \leq \Omega_{e}, ~\forall k, \label{P6-3} \\
			& && \tilde{\varphi}_k \geq  \xi_k \sqrt{\widehat{V}_{k}}, ~\forall \|\Delta\pmb{u}_k\|_2 \leq \Omega_{k}, ~\forall k, \label{P6-4}\\
			& && \tilde{\varphi}_{k_e} \geq \xi_{e} \sqrt{\widehat{V}_{e}}, ~\forall \|\Delta\pmb{u}_e\|_2 \leq \Omega_{e}, ~\forall k, \label{P7-5x} \\
			& && \eqref{P2-2}, ~\eqref{P2-3}, ~\eqref{P2-4},~\eqref{P2-5}. \label{P7-6x}
		\end{alignat}
	\end{subequations}
	One can notice that the constraint \eqref{P2-4} in problem $(\mathcal{P}2)$ has been captured by the constraints \eqref{P6-2}, \eqref{P6-3}, \eqref{P6-4}, and \eqref{P7-5x} in problem $(\mathcal{P}2.3)$ (referring to the CSI factor in the channel definition \eqref{chaCSIdef}).  Next,  we employ the SCA technique and the $\mathcal{S}$-procedure to transform the semi-infinite non-convex constraints \eqref{P6-2}, \eqref{P6-3}, \eqref{P6-4}, and \eqref{P7-5x} to finite LMIs. 
	Then, we leverage a PCCP algorithm to tackle the UMC \eqref{P2-3} \cite{lipp2016variations}. 
	\subsubsection{Sub-Problem for Optimizing the Beamforming Vectors} \label{subsec3}
	To linearize the semi-infinite inequalities in \eqref{P6-2}, we first substitute \eqref{ckedef} into  \eqref{P6-2}, hence, \eqref{P6-2} can be expressed as:
	\begin{align} \label{P3CSI-BnLx2}
		2^{\varphi_k}-1 \leq \frac{\left|(\pmb{u}_k \pmb{\Phi} \pmb{G}_{_\text{AR}})\pmw_{k}\right|^2}{\left\|(\pmb{u}_k \pmb{\Phi} \pmb{G}_{_\text{AR}})\pmW_{-k}\right\|^2 +\sigk}, \forall k.
	\end{align}
	By treating the interference plus noise signal as an auxiliary function $\pmb{\beta} = [\beta_1, \dots, \beta_{K}]$, \eqref{P3CSI-BnLx2} can be expressed as: 
	\begin{subequations}
		\begin{equation} \label{P3CSI-BLx3}
			\left|(\pmb{u}_k \pmb{\Phi} \pmb{G}_{_\text{AR}})\pmw_{k}\right|^2 \geq (2^{\varphi_k}-1)\beta_{k}, \forall k,
		\end{equation}
		\begin{equation} \label{P3CSIn-BLx2}
			\left\|(\pmb{ u}_k \pmb{\Phi} \pmb{G}_{_\text{AR}})\pmW_{-k}\right\|^2 +\sigk \leq \beta_{k}, \forall k.
		\end{equation}
	\end{subequations}
	To circumvent the non-convex semi-infinite inequalities in \eqref{P3CSI-BLx3}, we replace the left-hand side of  \eqref{P3CSI-BLx3} with its lower bounds using the following Lemma.

	\begin{mylem} \label{lem1_1}
		At iteration $(\iota)$, let $\wk$ and $\thetak$ be the optimal solution, then at the point $(\wk,\thetak)$ we can express the linear lower bound of  \eqref{P3CSI-BLx3} as:
		\begin{align} \label{ICSIeqw}
			\left|(\pmb{u}_k \pmb{\Phi} \pmb{G}_{_\text{AR}})\pmw_k\right|^2 \triangleq \pmb{u}_k \mathbf{X}_k \pmb{u}_k^H, 
		\end{align}
		where $\mathbf{X}_k$, is given in \eqref{LMI1-1}.
	\end{mylem}
	
	\Prf Refer to Appendix \ref{apendixB}. 
	
	Next, by substituting \eqref{ICSIeqw} in  \eqref{P3CSI-BLx3}, and using \eqref{ch2-h1} and Lemma \ref{lem1_1}, the inequality in \eqref{P3CSI-BLx3} is reformulated as in \eqref{Bobx1}, where $d_k \triangleq \widehat{\pmb{u}}_{k} \mathbf{X}_{k} \widehat{\pmb{u}}_{k}^H$.

	To address the uncertainty of $\{\Delta\pmb{u}_k\}$ in \eqref{Bobx1}, we leverage the $\mathcal{S}$-procedure \cite{boyd1994linear}. 
	\begin{mylem} \label{lem1_2}
		($\mathcal{S}$-procedure) For any Hermitian matrix $\mathbf{U}_i \in \mathbb{C}^{L \times L}$, vector  $\mathbf{u}_i \in \mathbb{C}^{L \times 1}$, and scalar $\text{u}_i$, for ${i}= 0,\dots,Q$. A quadratic function of a variable $x$ is defined as: \setcounter{equation}{43}
		\begin{align}
			f_{i}(x) \triangleq x^H \mathbf{U}_{i} x + 2\Re\{\mathbf{u}_{i}^H x\}+{\text{u}}_{i}.
		\end{align}
		The condition $f_{0}(x) \geq 0$ such that $f_{i}(x) \geq 0, ~ i = 1,\dots,Q$,  holds, if an only if there exists $\mathfrak{n}_{i} \geq 0, i= 0,\dots,Q$, such that,
		\begin{align} 
			\begin{bmatrix}
				\mathbf{U}_{0} & \mathbf{u}_{0} \\ \mathbf{u}_{0}^H & \text{u}_{0}	
			\end{bmatrix} - \sum\limits_{i= 0}^{Q} \mathfrak{n}_{i} \begin{bmatrix}
				\mathbf{U}_{i} & \mathbf{u}_{i} \\ \mathbf{u}_{i}^H & \text{u}_{i}	
			\end{bmatrix} \succeq 0.
		\end{align} 
	\end{mylem}
	
	Using Lemma \ref{lem1_2}, we can transform \eqref{Bobx1} into its equivalent LMIs as:
	\begin{align} \label{P3CSI-BLX3n2}
		\begin{bmatrix}
			\varpi_{k} \mathbf{I}_{M}+ \mathbf{X}_k & (\widehat{\pmb{u}}_{k} \mathbf{X}_k)^H \\ (\widehat{\pmb{u}}_{k} \mathbf{X}_k) & d_k - (2^{\varphi_k}-1)\beta_{k} - \varpi_{k} \Omega_{k}^2	
		\end{bmatrix} \succeq 0, \forall k,
	\end{align}
	where $\pmb{\eta} \triangleq [\eta_{1},\dots,\eta_{K}]^T \geq 0$ are slack variables. Even though \eqref{P3CSI-BLx3} has been transformed into an LMI form, it is still non-convex due to the non-convex nature of the term $2^{\varphi_k} \beta_{k}$. For that, we adopt the SCA technique \cite{9197675} to convert the non-convex constraint \eqref{P3CSI-BLx3} to a convex approximation expression. Specifically, performing the first order Taylor approximation, the term $2^{\varphi_k} \beta_{k}$ can be upper bounded by:
	\begin{align} \label{upbeta}
		(\beta_{k}(2^{\varphi_k}))^{ub} \triangleq ((\varphi_k-\varphi_k^{(\iota)})(\beta_{k}^{(\iota)}) \ln2 + \beta_{k})2^{\varphi_k^{(\iota)}},
	\end{align}		
	where $\varphi_k^{(\iota)}, \beta_{k}^{(\iota)}$ are the value of the variables $\varphi_k, \beta_{k}$ at iteration $(\iota)$ in the SCA-based algorithm, respectively. 
    
    Lastly, substituting \eqref{upbeta} in \eqref{P3CSI-BLX3n2}, the LMIs in \eqref{P3CSI-BLX3n2} can be expressed as:
	\begin{align} \label{P3CSI-BLX3n3}
		\begin{bmatrix}
			\varpi_{k} \mathbf{I}_{M}+ \mathbf{X}_k & (\widehat{\pmb{u}}_{k} \mathbf{X}_k)^H \\ (\widehat{\pmb{u}}_{k} \mathbf{X}_k) & d_k - (\beta_{k}(2^{\varphi_k}))^{ub}+\beta_{k} - \varpi_{k} \Omega_{k}^2	
		\end{bmatrix} \succeq 0, \forall k.
	\end{align} 
	
	Next, we tackle  \eqref{P3CSIn-BLx2} using Schur's complement \cite{boyd2004convex}.
	\begin{mylem} \label{lem1_3}
		(Schur's complement) For given matrices $\mathbf{U} \succeq 0$, $\mathbf{Y}$, and $\mathbf{Z}$, let a Hermitian matrix $\mathbf{X}$ be defined as: 
		\begin{align} 
			\mathbf{X}\triangleq
			\begin{bmatrix}
				\mathbf{Z} & \mathbf{Y}^H \\ \mathbf{Y} & \mathbf{U}
			\end{bmatrix}.
		\end{align} 
	\end{mylem}
	Then, $\mathbf{X} \succeq 0$ if and only if $\Delta{\mathbf{U}} \succeq 0$, where $\Delta{\mathbf{U}}$ is the Schur's complement define as $\Delta{\mathbf{U}} \triangleq \mathbf{Z}-\mathbf{Y}^H\mathbf{U}^{-1}\mathbf{Y}$.
	
	Using Lemma \eqref{lem1_3}, we can equivalently recast \eqref{P3CSIn-BLx2} as:
	\begin{align} \label{BobLMIx2n}
		\begin{bmatrix}
			\beta_{k} & \mathbf{t}_k^H \\ \mathbf{t}_{k}^H & \mathbf{I}_{K-1}
		\end{bmatrix} \succeq 0, \forall \|\Delta\pmb{u}_k\|_2 \leq \Omega_{k}, \forall k,
	\end{align}
	where $\mathbf{t}_{k} \triangleq \left((\widehat{\pmb{u}}_{k}^H \mathbf{\Phi} \pmb{G}_{_\text{AR}})\mathbf{W}_{-k}\right)^H$, and $\mathbf{W}_{-k} \triangleq [\pmw_1, \dots,\pmw_{k-1},\pmw_{k+1}, \dots,\pmw_{\mathcal{K}}] \in \mathbb{C}^{M \times \mathcal{K}-1}$.
	Next, we use Nemirovski's Lemma \cite{1369660} to further handle \eqref{BobLMIx2n}.
	\begin{mylem} \label{lem1_4}
		(Nemirovski's Lemma) For any Hermitian matrix $\mathbf{A}$, matrices $\mathbf{B}$, $\mathbf{C}$, and $\mathbf{X}$, and scalar $a$, the following LMI holds,
		\begin{align}
			\mathbf{A} \succeq \mathbf{B}^H \mathbf{X} \mathbf{C}+\mathbf{C^H \mathbf{X}^H} \mathbf{B}, ~\text{for} \|\mathbf{X}\| \leq a,
		\end{align}
		if and only if 
		\begin{align} 
			\begin{bmatrix}
				\mathbf{A}- a\mathbf{C}^H \mathbf{C} & -t \mathbf{B}^H \\ -t\mathbf{B} & a \mathbf{I}
			\end{bmatrix} \succeq 0,
		\end{align} 
	\end{mylem}
	where $a$ is a non-negative real number.
	
	Using Lemma \eqref{lem1_4} and introducing the slack variable  $\pmb{\varkappa} \triangleq [\varkappa_1, \dots, \varkappa_k] \geq 0$, \eqref{BobLMIx2n} is recast as: 
	\begin{align} \label{P3CSI-BLX4n}
		\begin{bmatrix}
			B_{11,k}& \widehat{t}^H_k & \mathbf{0}_{1 \times M }  \\
			\widehat{t}_k & \mathbf{I}_{K-1} & \Omega_{k} (\mathbf{\Phi} \pmb{G}_{_\text{AR}} \pmW_{-k})^H  \\  \mathbf{0}_{M \times 1 } & \Omega_{k} (\mathbf{\Phi} \pmb{G}_{_\text{AR}} \pmW_{-k}) & \varkappa_k \mathbf{I}_{M} 
		\end{bmatrix} \succeq 0, \forall k,
	\end{align}
	where $B_{11,k} \triangleq \beta_k - \sigk - \varkappa_k$, and $\widehat{t}_k \triangleq \left((\widehat{\pmb{u}}_{k} \mathbf{\Phi} \pmb{G}_{_\text{AR}})\mathbf{W}_{-k}\right)^H$.
	
	Next, we tackle \eqref{P6-3} by firstly substituting \eqref{ckedef} into \eqref{P6-3}, and then treat the interference plus noise signal in  \eqref{P6-3} as an auxiliary function $\pmb{\beta}_{e_k}\triangleq [\beta_{e_1},\dots,\beta_{e_k}]$. We hence can recast \eqref{P6-3} as
	\begin{subequations}
		\begin{alignat} {2}
			&\left|(\pmb{u}_{e} \mathbf{\Phi} \pmb{G}_{_\text{AR}})\pmw_{k}\right|^2 \leq (2^{\mu_{k_e}}-1)\beta_{ke}, \forall k,\label{EveLMI1} \\
			&\left\|(\pmb{u}_{e} \mathbf{\Phi} \pmb{G}_{_\text{AR}})\pmW_{-k}\right\|^2 +\sige \geq \beta_{ke}, \forall k. \label{EveLMI2}
		\end{alignat}
	\end{subequations}
	
	To tackle the uncertainties of $\{\Delta\pmb{u}_e\}$ in the constraints \eqref{EveLMI1} and \eqref{EveLMI2}, we use a similar approach as in \eqref{P3CSIn-BLx2}. Hence, the equivalent LMIs of \eqref{EveLMI1} and \eqref{EveLMI2} are, respectively, given by,	
	\begin{alignat}{2}
		\begin{bmatrix}
			C_{e,k}& \widehat{c}_{ke}^H & \mathbf{0}_{1 \times M }  \\
			\widehat{c}_{ke} & 1 & \Omega_e (\mathbf{\Phi} \pmb{G}_{_\text{AR}} \pmw_k)^H  \\  \mathbf{0}_{M \times 1 } & \Omega_{e} (\mathbf{\Phi} \pmb{G}_{_\text{AR}} \pmw_k) & \vartheta_{k} \mathbf{I}_{M} 
		\end{bmatrix} \succeq 0, \forall k, \label{EveLMI11} \\
		\begin{bmatrix}
			C_{11,k}& -\widehat{c}_{ke}^H & \mathbf{0}_{1 \times M }  \\
			-\widehat{c}_{ke} & \mathbf{I}_{K-1} & -\Omega_e (\mathbf{\Phi} \pmb{G}_{_\text{AR}} \pmW_{-k})^H  \\  \mathbf{0}_{M \times 1 } & -\Omega_{e} (\mathbf{\Phi} \pmb{G}_{_\text{AR}} \pmW_{-k}) & \varrho_{k} \mathbf{I}_{M} 
		\end{bmatrix} \succeq 0, \forall k, \label{EveLMI22}
	\end{alignat}
    
	where   $\pmb{\vartheta} \triangleq [\vartheta_{1}, \dots, \vartheta_{K}] \geq 0$ and $\pmb{\varrho} \triangleq [\varrho_{1}, \dots, \varrho_{K}] \geq 0$ are a slack variables,     
	\begin{align*}
		C_{e,k} & \triangleq (\beta_{ke}(2^{\mu_{k_e}}))^{ub}-\beta_{ke}-\vartheta_{k}, \\
		(\beta_{ke}(2^{\mu_{k_e}}))^{ub} & \triangleq ((\mu_{k_e}-\mu_{k_e}^{(\iota)})(\beta_{ke}^{(\iota)}) \ln2 + \beta_{ke})2^{\mu_{k_e}^{(\iota)}}, \\
		\widehat{c}_{ke} &\triangleq ((\widehat{\pmb{u}}_{e}^H \mathbf{\Phi} \pmb{G}_{_\text{AR}})\mathbf{W}_{-k})^H, \\
		C_{11,k} &\triangleq \beta_{ke} - \sige - \varrho_{k}.
	\end{align*}
	% Next, we show the steps to derive the LMI of equations \eqref{P6-4} and \eqref{P7-5x}. Using \eqref{dispfa2},  the inequalities \eqref{xy} and \eqref{xt} in Appendix \ref{AppA}, and by defining $x\triangleq V_i$, $\mathbf{A} \triangleq \rho_{i}$, and $B \triangleq \varsigma_i$, the dispersion factors for $i \in \{k,e\}$ can be expressed as:
    Next, we show the steps to derive the LMI of equations \eqref{P6-4} and \eqref{P7-5x}. Using \eqref{dispfa2}, the dispersion factors are expressed as functions of the variables $(\widehat{\rho}_i,\widehat{\upsilon}_i)$. Then, by applying the inequalities \eqref{xy} and \eqref{xt} in Appendix \ref{AppA}, the dispersion factors for $i \in \{k,e\}$ can be expressed as: 
	\begin{align} 
		\xi_{i}\sqrt{V_{i}}& \leq \frac{\xi_{i} \sqrt{V_{i}^{(\iota)}}}{2}+ \frac{\xi_{i} |(\pmb{u}_i \pmb{\Phi} \pmb{G}_{_\text{AR}})\pmw_k|^2}{\sqrt{V_{i}^{(\iota)}}\left(\|(\pmb{u}_i \pmb{\Phi} \pmb{G}_{_\text{AR}}) \pmW_{k}\|_2^2+\textcolor{black}{\sigma_i^2}\right)},
	\end{align}
	where $\mathbf{W}_k \triangleq [\pmw_{1},\dots,\pmw_{\mathcal{K}}] \in \mathbb{C}^{M \times \mathcal{K}}$. 
	Hence, we can express \eqref{P6-4} and \eqref{P7-5x} as, 
	\begin{align} \label{P3CSI-BFx2x}
		\frac{(\tilde{\varphi}_{i}-\mathfrak{L})}{\mathfrak{Y}} \geq    \frac{|(\pmb{u}_{i} \pmb{\Phi} \pmb{G}_{_\text{AR}})\pmw_{k}|^2}{\|(\pmb{u}_{i} \pmb{\Phi} \pmb{G}_{_\text{AR}}) \pmW_{k}\|_2^2+\sigma_{i}}, ~i \in \{k,e\}, 
	\end{align}
	where $\mathfrak{L} \triangleq \frac{\xi_{i} \sqrt{V_{i}^{(\iota)}}}{2}$, and $\mathfrak{Y} \triangleq {\xi_{i}}/{\sqrt{V_{i}^{(\iota)}}}$. 
	
	Similar to \eqref{P3CSI-BnLx2}, by treating the interference plus noise signal as an auxiliary function $\pmb{\zeta} \triangleq [\zeta_{1}, \dots, \zeta_{K}]$, \eqref{P3CSI-BFx2x} can be expressed as:
	\begin{subequations}
		\begin{alignat} {2}
			&\left|(\pmb{u}_i \pmb{\Phi} \pmb{G}_{_\text{AR}})\pmw_{k}\right|^2  \leq \left({(\tilde{\varphi}_{i}-\mathfrak{L})}/{\mathfrak{Y}}\right) \zeta_{i}, \forall k, i \in \{k,e\}, \label{BobFLMI1}\\
			&\left\|(\pmb{ u}_{i} \pmb{\Phi} \pmb{G}_{_\text{AR}})\pmW_{k}\right\|^2 +\sigma_{i} \geq \zeta_{k}, \forall k, i \in \{k,e\}. \label{BobFLMI2}
		\end{alignat}
	\end{subequations}
	Similar to \eqref{EveLMI22}, we can express \eqref{BobFLMI2} by its equivalent LMI  as: 
	\begin{align} \label{BobFLMI22}
		\begin{bmatrix}
			\widetilde{C}_{11,i}& -\tilde{c}^H_i & \mathbf{0}_{1 \times M }  \\
			-\tilde{c}_i & \mathbf{I}_{K} & -\Omega_i (\mathbf{\Phi} \pmb{G}_{_\text{AR}}) \pmW_{k})^H  \\  \mathbf{0}_{M \times 1 } & -\Omega_{i} (\mathbf{\Phi} \pmb{G}_{_\text{AR}}) \pmW_{k}) & \mathfrak{u}_{i} \mathbf{I}_{M} 
		\end{bmatrix} \succeq 0, j\neq k,
	\end{align}
	where $\widetilde{C}_{11,i} \triangleq \zeta_{k} - \sigk - \mathfrak{u}_{i}$,  $\mathfrak{u} \triangleq [\mathfrak{u}_{1}, \mathfrak{u}_{2}, \dots, \mathfrak{u}_{i}] \geq 0$ is a slack variable, and $\tilde{c}_i \triangleq ((\widehat{\pmb{u}}_{i}^H \mathbf{\Phi} \pmb{G}_{_\text{AR}}))\mathbf{W}_{k})^H$.

	Next, we can express \eqref{BobFLMI1} as the following LMI:
	\begin{align} \label{BobFLMI11x}
		\begin{bmatrix}
			C_{k_f}& \tilde{c}^H_{ke} & \mathbf{0}_{1 \times M }  \\
			\tilde{d}_{i} & 1 & \Omega_e (\mathbf{\Phi} \pmb{G}_{_\text{AR}}) \pmw_k)^H  \\  \mathbf{0}_{M \times 1 } & \Omega_{e} (\mathbf{\Phi} \pmb{G}_{_\text{AR}}) \pmw_k) & \tilde{\mathfrak{u}}_{i} \mathbf{I}_{M} 
		\end{bmatrix} \succeq 0,  j\neq k,
	\end{align}
	where $C_{k_f} \triangleq \frac{\tilde{\varphi}_{k}\zeta_{k}}{\mathfrak{Y}}-\frac{\mathfrak{L}\zeta_{k}}{\mathfrak{Y}} -\tilde{\mathfrak{u}}_{i}$,  $\tilde{\pmb{\mathfrak{u}}} \triangleq [\tilde{\mathfrak{u}}_{1}, \dots, \tilde{\mathfrak{u}}_{K}] \geq 0$ is a slack variable, and $\tilde{d}_{i} \triangleq ((\widehat{\pmb{u}}_{e} \mathbf{\Phi} \pmb{G}_{_\text{AR}})\pmw_{k})^H$.
	
	However, \eqref{BobFLMI11x} is still non-convex due to the non-convex nature of the term $\tilde{\varphi} \zeta_{k}$. Similar to \eqref{P3CSI-BLX3n2}, we employ the SCA technique to tackle the non-convex expression as: 
	\begin{align} \label{upperterm2}
		\left(\tilde{\varphi} \zeta_{i}\right)^{ub} \triangleq \tilde{\varphi}^{(\iota)} \zeta_{i} + (\tilde{\varphi}-\tilde{\varphi}^{(\iota)})\zeta_{i}^{(\iota)},
	\end{align}
	where $\tilde{\varphi}^{(\iota)}$, $\zeta_{i}^{(\iota)}$ are the value of the variables $\tilde{\varphi}_{i}$, $\zeta_{i}$ at iteration $(\iota)$ in the SCA-based algorithm, respectively. Be substituting \eqref{upperterm2} in \eqref{BobFLMI11x}, we can express \begin{algorithm}
		\caption{PCCP-Based Algorithm for Solving Problem $(\mathcal{P}2.5)$}
		\label{alg3}
		\begin{algorithmic}[1]
			\State \textbf{Input:} Initial values $(\pmb{w}^{(1)}, \btheta^{(1)})$, maximum penalty $o_{\text{max}}$, scaling factor $\nu \geq 1$, tolerances $\epsilon_{t_1}, \epsilon_{t_2}$, maximum iterations $\iota_{\text{max}}$.
			\State \textbf{Initialize:} Set iteration index $\iota = 1$.
			\Repeat
			\State Solve $(\mathcal{P}2.3)$ to obtain $\btheta^{(\iota)}$.
			\If{$\left\|e^{(j \thetako)} - e^{(j \thetak)}\right\|_1 \leq \epsilon_{t_1}$ \textbf{and} $Q \leq \epsilon_{t_2}$}
			\State \textbf{Output:} $\btheta^* = \btheta^{(\iota)}$; \textbf{terminate}.
			\Else
			\State Update $o^{(\iota+1)} = \min\{\nu o^{(\iota)}, o_{\text{max}}\}$;
			\State Increment $\iota \gets \iota + 1$.
			\EndIf
			\Until{$\iota > \iota_{\text{max}}$}
			\State \textbf{Restart:} Redefine $\btheta^{(1)}$, choose $\nu > 1$, set $\iota = 1$, and repeat.
			\State \textbf{Output:} Final feasible solution $\btheta^* = \btheta^{(\iota)}$.
		\end{algorithmic}
	\end{algorithm}
    \eqref{BobFLMI11x} as:
	\begin{align} \label{BobFLMI11}
		\begin{bmatrix}
			\widetilde{C}_{k_f}& \tilde{c}^H_{ke} & \mathbf{0}_{1 \times M }  \\
			\tilde{d}_{i} & 1 & \Omega_{e} (\mathbf{\Phi} \pmb{G}_{_\text{AR}}) \pmw_k)^H  \\  \mathbf{0}_{M \times 1 } & \Omega_{e} (\mathbf{\Phi} \pmb{G}_{_\text{AR}}) \pmw_k) & \tilde{\mathfrak{u}}_{i} \mathbf{I}_{M} 
		\end{bmatrix} \succeq 0,  j\neq k,
	\end{align}
	where $\widetilde{C}_{k_f} \triangleq \frac{\left(\tilde{\varphi} \zeta_{k}\right)^{ub}}{\mathfrak{Y}}-\frac{\mathfrak{L}\zeta_{k}}{\mathfrak{Y}} -\tilde{\mathfrak{u}}_{i}$.

	Eventually, reformulating problem $(\mathcal{P}2.3)$ with \eqref{P3CSI-BLX3n3}, \eqref{P3CSI-BLX4n}, \eqref{EveLMI11}, \eqref{EveLMI22}, \eqref{BobFLMI22}, and \eqref{BobFLMI11} yields the following optimization problem:
	\begin{subequations}\label{P6.3}
		\begin{alignat}{2}
			&(\mathcal{P}2.4):  && ~\underset{\pmw}{\max} ~ z, \label{P6.3-1} \\
			&~ \text{s.t.}       && z \leq {\sum}_{k}^{\mathcal{K}}{\varphi}_k - \mu_{k_e} - \tilde{\varphi}_k - \tilde{\varphi}_{k_e}, \forall k, \\
			& && \eqref{P3CSI-BLX3n3}, \eqref{P3CSI-BLX4n}, \eqref{EveLMI11}, \eqref{EveLMI22}, \eqref{BobFLMI22}, \eqref{BobFLMI11}, i \in \{k,e\}, \\
			& && \pmb{\varpi}\geq 0, \pmb{\varkappa} \geq 0, \pmb{\vartheta} \geq 0, \pmb{\varrho} \geq 0, \mathfrak{u} \geq 0, \tilde{\mathfrak{u}} \geq 0, \label{auxvarcon}\\
			& && \eqref{P2-2}, ~\eqref{P2-4},~\eqref{P2-5}.
		\end{alignat}
	\end{subequations}
	At this point, all nonlinear constraints \eqref{P6-2}, \eqref{P6-3}, \eqref{P6-4}, and \eqref{P7-5x} have been linearized as LMIs \eqref{P3CSI-BLX3n3}, \eqref{P3CSI-BLX4n}, \eqref{EveLMI11}, \eqref{EveLMI22}, \eqref{BobFLMI22}, and \eqref{BobFLMI11}; thus, problem ($\mathcal{P}2.4$) is an SDP problem that can be efficiently solved using standard solvers, e.g., the interior point method or the CVX toolbox \cite{grant2014cvx}. 
	\subsubsection{Sub-Problem for Optimizing the PREs} \label{subsec4}
	For a given $\pmw$, we aim to find $\thetako$ such that, $\widehat{R}_{s}(\wko_k,\thetako) \geq \widehat{R}_{s}(\wko_k,\thetak)$. Similar to the analysis in the previous section, we can express $\eqref{P6-2}$ and $\eqref{P6-3}$ with their equivalent LMIs \eqref{P3CSI-BLX3n3}, \eqref{P3CSI-BLX4n}, \eqref{EveLMI11}, and \eqref{EveLMI22}, respectively. To tackle the UMC \eqref{P2-3}, we leverage the PCCP method \cite{9525400}. The PCCP's main idea is to add slack variables to relax the problem so that if the UMC is violated, we can then penalize the sum of violations. The converged PCCP solution is an approximate first-order optimal solution of the original problem \cite{9525400}. 
	
	To apply the PCCP method, we introduce an auxiliary variable set $\mathbf{b} \triangleq \{b_{n}|n \in N\}$ satisfying $b_n \triangleq |e^{(j \theta_n)}|^* |e^{(j \theta_n)}|$. Then, \eqref{P2-3} can be expressed as $b_n \leq |e^{(j \theta_n)}|^* |e^{(j \theta_n)}| \leq b_n $. The non-convex part $b_n \leq |e^{(j \theta_n)}|^* |e^{(j \theta_n)}|$ can be approximated by $b_n \leq 2 \Re \{|e^{(j \theta_n)}|^* |e^{(j \theta^{(\iota)}_n)}| - |e^{(j \theta^{(\iota)}_n)}|^* |e^{(j \theta^{(\iota)}_n)}|\}$. Following the PCCP framework, we penalize the objective function \eqref{P6-1}, hence, problem  $(\mathcal{P}2.4)$ can be recast as:
    \begin{algorithm}
		\caption{Proposed AO Algorithm for Solving Problem $(\mathcal{P}2)$}
		\label{alg4}
		\begin{algorithmic}[1]
			\State \textbf{Initialize:} $(\pmb{w}^{(1)} = \pmw^{*}, \btheta^{(1)} = \btheta^{*}$, ~$\varsigma_{\max}$,~$\varsigma_{e}^{\max}$,~$T_{\max}$), set iteration index $\iota = 1$, and convergence tolerance $\epsilon_t > 0$.
			\Repeat
			\State Update $\pmb{w}^{(\iota+1)}$ by solving problem $(\mathcal{P}2.1)$;
			\State Update $\btheta^{(\iota+1)}$ by solving problem $(\mathcal{P}2.2)$;
			\If{$\dfrac{|R_s(\pmb{w}^{(\iota+1)}, \btheta^{(\iota+1)}) - R_s(\pmb{w}^{(\iota)}, \btheta^{(\iota)})|}{R_s(\pmb{w}^{(\iota)}, \btheta^{(\iota)})} \leq \epsilon_t$}
			\State Set $\pmw^* \gets \pmb{w}^{(\iota+1)}$, $\btheta^* \gets \btheta^{(\iota+1)}$;
			\State \textbf{Output:} $(\pmw^*, \btheta^*)$; \textbf{terminate}.
			\Else
			\State $\iota \gets \iota + 1$
			\EndIf
			\Until{convergence}
		\end{algorithmic}
	\end{algorithm}
	\begin{subequations} \label{P3CSIthBL}
		\begin{alignat} {2}
			&(\mathcal{P}2.5): && \underset{\btheta}{\max}  ~ z - o Q, \label{P3CSIthBL-1} \\
			&~ \text{s.t.}  && \eqref{P3CSI-BLX3n3}, ~\eqref{P3CSI-BLX4n}, ~\eqref{EveLMI11}, ~\eqref{EveLMI22}, ~\eqref{BobFLMI22}, ~\eqref{BobFLMI11}, ~\eqref{P2-4} \\
			& && |e^{(j \theta_n)}|^* |e^{(j \theta_n)}| \leq b_n + c_n, \label{P3CSIthBL-BLx1} \\
			& && b_n - \hat{c}_n \leq 2 \Re \left\{|e^{(j \theta_n)}|^* |e^{(j \theta^{(\iota)}_n)}|\right\}  \nonumber \\ 
			& &&- |e^{(j \theta^{(\iota)}_n)}|^* |e^{(j \theta^{(\iota)}_n)}|, \label{P3CSIthBL-BLx2} \\
			& &&  b_n \geq 0, \forall n \in N, \label{P3CSIthBL-BLx3}
		\end{alignat}
	\end{subequations}
	where $Q \triangleq {\sum}_{n=1}^N c_n+\hat{c}_n$ is the penalty term, $\mathbf{c} \triangleq \{c_n,\hat{c}_n\}$ is the slack variable imposed over the modulus constraint \eqref{P2-3}, and $o$ is the regularization factor. The regularization factor is imposed to control the UMC constraint \eqref{P2-3} by scaling the penalty term $Q$.
	
	Problem $(\mathcal{P}2.5)$ is an SDP problem that can be solved by the CVX toolbox. Unlike the conventional SDR method to tackle UMC, the PCCP method, summarized in Algorithm \ref{alg3}, is guaranteed to find a feasible point for problem $(\mathcal{P}2.5)$ \cite{9774882}.  
	The following points are helpful for the numerical implementation of the Algorithm \ref{alg3}:
	\begin{enumerate}[label=(\alph*)]
		\item We invoke $o_{\text{max}}$ to avoid numerical complications if $o$ grows too large\cite{9505311};
		\item The stopping criteria $Q \leq \epsilon_{t_2}$ guarantees the UMC \eqref{P2-3} when $\epsilon_{t_2}$ is small \cite{9180053};
		\item The convergence of Algorithm \ref{alg3} is controlled by $\left\|e^{(j \btheta^{(\iota+1)})}-e^{(j \btheta^{(\iota)})}\right\|_1 \leq \epsilon_{t_1}$.
	\end{enumerate}
	
	Finally, the pseudo code of the AO algorithm to solve problem ($\mathcal{P}2$) is shown in Algorithm \ref{alg4}. This algorithm converges to a locally optimal solution of ($\mathcal{P}2$), which can be proved in the following theorem.
	\begin{theorem} \label{theo2}
		Algorithm \ref{alg4} generates a locally optimal solution for problem $(\mathcal{P}2)$.
	\end{theorem}
	\textit{Proof}: See Appendix \ref{AppD}.
	\subsubsection{Complexity Analysis}
	Since convex problems {$(\mathcal{P}2.1)$} and {$(\mathcal{P}2.2)$} involve linear constraints and LMIs {$\eqref{P3CSI-BLX3n3}, \eqref{P3CSI-BLX4n}, \eqref{EveLMI11}, \eqref{EveLMI22}, \eqref{BobFLMI22}$ and $\eqref{BobFLMI11},$} they can be solved by the interior point method\cite{grant2014cvx}. The algorithm's complexity is defined by its worst-case runtime and the number of operations \cite{9180053}. For problem {$(\mathcal{P}2.1)$}, the number of variables is ${c}\triangleq 2M$, the size of \eqref{P3CSI-BLX3n3} is ${a_1}\triangleq MN+M+1$, the size of \eqref{P3CSI-BLX4n}, and \eqref{EveLMI22} is ${a_2}\triangleq 2M+1$, and the size of  \eqref{EveLMI11} is ${a_3}\triangleq 2M+1$. Thus, the complexity of solving problems {$(\mathcal{P}2.1)$} and {$(\mathcal{P}2.2)$} can be, respectively, calculated as follows:
    \begin{algorithm}
		\caption{Proposed AO Algorithm for Solving Problem $(\mathcal{P}3)$} \label{alg7}
		\begin{algorithmic}[1]
			\State \textbf{Initialize:} $\pmb{\text{w}}^{(1)}$,~ $\btheta^{(1)}$,~$\varsigma_{\max}$,~$\varsigma_{e}^{\max}$,~$T_{\max}$, convergence tolerance $\epsilon_t > 0$,  and Set $\iota=1$. 
			\State \textbf{Repeat} 
			\State Update $\wko$ by solving problem $(\mathcal{P}2.2)$, and $\thetako_n$ using the PCCP algorithm;
			\State \textbf{if} $\dfrac{|P_T^{(\iota+1)}-P_T^{(\iota)}|}{P_T^{(\iota)}} \leq \epsilon_t$.
			\State \textbf{Then} $\thetak\leftarrow  \thetako_n $, $\wk\leftarrow \wko$ and terminate. 
			\State \textbf{Otherwise} $\iota\leftarrow \iota+1$ and continue.
			\State \textbf{Output} $(\wk_k, \thetak)$.
		\end{algorithmic}
	\end{algorithm}
	\begin{align}
		\mathcal{O}_{\pmw} &\triangleq \mathcal{O}\left\{\sqrt{{b_1}+2} ({c})({c}^2+{c} {b_2}+{b_3}+({c}+1)^2)\right\}, \label{comx1}\\ 
		\mathcal{O}_{\pmb{\Phi}} &\triangleq \mathcal{O}\left\{\sqrt{{b_1}+4N} (2N)(4N^2+2N {b_2}+{b_3}+4NM)\right\} \label{comx2},
	\end{align}
	where 
	\begin{align*}
		{b_1} & \triangleq {\sum}_{k}^{\mathcal{K}}({a_1}+{a_2})+2k({a_2}+{a_3})+(k-2)({a_1}+{a_3}), \\
		{b_2} & \triangleq {\sum}_{k}^{\mathcal{K}}({a_1}^2+{a_2}^2)+2k({a_2}^2+{a_3}^2)+(k-2)({a_1}^2+{a_3}^2), \\
		{b_3}&\triangleq{\sum}_{k}^{\mathcal{K}}({a_1}^3+{a_2}^3)+2k({a_2}^3+{a_3}^3)+(k-2)({a_1}^3+{a_3}^3).
	\end{align*}
    
    \section{Dealing with Eve's Unknown CSI} \label{sec5}
    In this scenario, we tackle the problem when Alice cannot obtain Eve's CSI.  To ensure secure communication, we aim to minimize the required transmit power $P_T$ while satisfying a QoS constraint, i.e., $\gamma_k \ge \gamma_e$. This approach enables the use of the remaining transmit power at Alice as AN to further degrade Eve's SINR $\gamma_{e}$. In contrast to existing works \cite{9146177,9764813}, which focus on a LBR setting, whereas our formulation considers an FBR environment characterized by inherently low latency and short transmission duration. Accordingly, the total transmit power $P_{tot}$ can be expressed as follows:
    \begin{align} \label{respo}
		P_{tot} \triangleq P_{T}+P_{J},
	\end{align}
	where the residual power $P_J\triangleq\sum_{k=1}^{\mathcal{K}}\mathbf{E}\{\|\pmb{\mathbf{Vr}}\|^2\}$ is used as AN, and $\mathbf{Vr}$ is the AN vector, while $\mathbf{V} \in \mathbb{C}^{M \times M-1}$ is the null space of the user's channel with $\pmb{h}_{k}(\btheta) \mathbf{V} = 0$, and $\mathbf{r} \in \mathbb{C}^{M-1}$ is a random vector whose elements are independent Gaussian random variables with zero mean and variance $\sigma_r \triangleq \frac{P_J}{M-1}$. Hence, we can express the SINR at user-$k$ as in \eqref{SINRex}, and Eve's eavesdropping rate as:
	\begin{align} \label{SINRAN}
		\widehat{\gamma}_e(\pmw,\btheta)&\triangleq\frac{|{\widehat{\pmb{h}}}_{e}(\btheta)\pmw_k|^2}{{\widehat{\rho}}_e+|\widehat{\pmb{h}}_{e}(\btheta)\mathbf{Vr}|^2}. 
	\end{align}
	To that end, we can formulate the problem as:
	\begin{subequations} \label{PAB}
		\begin{alignat}{2}
			&(\mathcal{P}3):  &&\underset{\pmw,\btheta}{\min} ~ P_T, \label{PAB-1} \\
			&~\text{s.t.}  &&\widehat{\gamma}_k(\pmw,\btheta)-\xi_k \sqrt{\widehat{V}_{k}} \geq \gamma_{th} \sigk, \label{PAB-4} \\
			& &&\eqref{P2-2}, \eqref{P2-3}, \eqref{P2-4}. \label{PAB-2}
		\end{alignat}
	\end{subequations}
	where $\gamma_{th}$ represents the QoS constraint at user-$k$. Due to the coupling of $\pmw$ and $\btheta$ in \eqref{PAB-4}, problem $(\mathcal{P}3)$ is a nonconvex problem. 
	\subsubsection{Sub-Problem for Optimizing the Beamforming Vectors} \label{subsec1-AN}
    To solve the newly formulated problem $(\mathcal{P}3)$ under unknown CSI, we adopt an AO-based approach, where the beamforming vector $\pmw$ and IRS phase shifts $\btheta$ are iteratively optimized. Under this setting, problem $(\mathcal{P}3)$ can be recast as:
\begin{subequations} \label{PAB02}
    \begin{alignat}{2}
        &(\mathcal{P}3.1):  &&\underset{\pmw}{\min} ~ P_T, \label{PAB02-1} \\
        &~\text{s.t.}  && \left|({\pmb{u}}_k \pmb{\Phi} {\pmb{G}_\text{AR}})\pmw_{k}\right|^2 \geq (2^{\gamma_k}-1)\hat{\beta}_{k},~ \forall k, \label{PAB02-3} \\
        & && \left\|({\pmb{u}}_k \pmb{\Phi} {\pmb{G}_\text{AR}})\pmW_{-k}\right\|^2 + \sigk \leq \hat{\beta}_{k},~ \forall k, \label{PAB02-4} \\
        & &&\left|(\pmb{u}_k \pmb{\Phi} \pmb{G}_{_\text{AR}})\pmw_{k}\right|^2  \leq \left(\frac{\tilde{\varphi}_{k}-\mathfrak{L}}{\mathfrak{Y}}\right) \hat{\beta}_{k},~ \forall k,  \label{BobFLMI1pt}\\
        & &&\left\|(\pmb{u}_{k} \pmb{\Phi} \pmb{G}_{_\text{AR}})\pmW_{k}\right\|^2 +\sigma_{k} \geq \hat{\beta}_{k},~ \forall k, \label{BobFLMI1pt2} \\
        & &&\eqref{P2-2}, {\eqref{P2-4}}, \eqref{P2-5}. \label{PAB02-5}
    \end{alignat}
\end{subequations}
where $\hat{\pmb{\beta}} \triangleq [\hat{\beta}_1, \dots, \hat{\beta}_{K}]$ denotes an auxiliary variable. Following the same procedure used to linearize constraints \eqref{P3CSI-BLx3}, \eqref{P3CSIn-BLx2}, \eqref{BobFLMI1}, \eqref{BobFLMI2}, and \eqref{BobFLMI2} in Section~\ref{sec4}, constraints \eqref{PAB02-3}, \eqref{PAB02-4}, \eqref{BobFLMI1pt}, and \eqref{BobFLMI1pt2} can be equivalently expressed as their LMIs given by \eqref{P3CSI-BLX3n3}, \eqref{P3CSI-BLX4n}, \eqref{BobFLMI22}, and \eqref{BobFLMI11}. Consequently, problem $(\mathcal{P}3.1)$ can be rewritten as:
\vspace{-0.1cm}
\begin{subequations} \label{PAB03}
    \begin{alignat}{2}
        &(\mathcal{P}3.2):  &&\underset{\pmw}{\min} ~ P_T, \label{PAB03-1} \\
        &~\text{s.t.}  && \eqref{P2-2}, \eqref{P2-4}, \eqref{P2-5}, \eqref{P3CSI-BLX3n3}, \eqref{P3CSI-BLX4n}, \eqref{BobFLMI22}, \eqref{BobFLMI11x}, \eqref{auxvarcon}.
    \end{alignat}
\end{subequations}
Problem $(\mathcal{P}3.2)$ is a convex problem and can be efficiently solved using the CVX toolbox.
	\subsubsection{Sub-Problem for Optimizing the PREs}
    Following a similar approach, our goal is to optimize $\thetako$ to minimize $P_T$ while ensuring the QoS constraint in \eqref{PAB-4}, the uncertainty constraint \eqref{P2-4}, and the FBR constraints \eqref{P2-5} are satisfied. This task can be interpreted as a search for a feasible $\btheta$ that simultaneously meets the QoS condition \eqref{PAB-4}, the user minimum constraint \eqref{P2-3}, and the FBR requirement \eqref{P2-5}, where the resulting solution may correspond to the optimal point. Subsequently, the residual power defined in \eqref{respo}, together with the PCCP described in Algorithm~\ref{alg3}, is utilized to compute the optimal $\btheta$.
    
	Finally, the algorithm to solve problem ($\mathcal{P}3$) is described in Algorithm \ref{alg7}. 
	\begin{theorem}
		The solution obtained using the AO Algorithm~\ref{alg7} is a locally optimal solution to problem $(\mathcal{P}3)$.
	\end{theorem}
	\textit{Proof}: See Appendix \ref{AppD}.
	\subsubsection{Complexity Analysis} Algorithm  \ref{alg7} involves linear constraints and LMIs $\eqref{P3CSI-BLX3n3}, \eqref{P3CSI-BLX4n}, \eqref{EveLMI11}, \eqref{EveLMI22}$, hence, its complexity is similar to Algorithm \ref{alg4}, which can be expressed as in \eqref{comx1} and \eqref{comx2}. 
    
	\section{Numerical Results} \label{sec6}
	\subsection{Parameter Setting}
	\begin{table}
\centering
\caption{Numerical Parameters}
\renewcommand{\arraystretch}{1.2}
\setlength{\tabcolsep}{6pt}

\begin{tabular}{| p{0.5\linewidth} | p{0.4\linewidth} |}
\hline
\textbf{Parameter} & \textbf{Numerical Value} \\
\hline

Noise power spectral density ($\sigma_i^2$) & -174 dBm/Hz \\
\hline
Alice transmit power ($P_T$) & 20 dBm \\
\hline
Antenna gains ($G_{\text{A}}, G_{\text{IRS}}$) & 5 dBi \\
\hline
Bandwidth ($\mathcal{B}$) & 1 MHz \\
\hline
Decoding error threshold ($\varsigma_{\max}$) & $10^{-5}$ \\
\hline
Information leakage threshold ($\varsigma_e^{\max}$) & $10^{-5}$ \\
\hline
Rician K-factor & 3 \\
\hline
Convergence tolerances ($\epsilon_t$, $\epsilon_{t_1}$, $\epsilon_{t_2}$) & $10^{-3}$ \\
\hline
Initialization parameters & $\beta_k^{(\iota)}=1$, $\beta_{ke}^{(\iota)}=1$, $o^{(1)}=10$, $o_{\max}=30$ \\
\hline

\end{tabular}
\label{T2}
\end{table}

	In this section, we perform an extensive simulation to evaluate the performance of the proposed approach. The results are obtained using MATLAB and CVX toolboxes.  In this setup, Alice is located at $(15,0,15)$, the IRS is located at $(0,25,40)$, and the users are randomly distributed to the right of Alice over a $(60m \times 60m)$ area. The Eve (e.g., a compromised legitimate user) is randomly located in $(100m \times 100m)$ outside the users' area. Note that if Eve is too close to one of the users, it is mainly impossible to guarantee the positive SR for all the users \cite{6772207}. In this case, other methods, e.g., encryption or friendly jamming, can provide users' secrecy \cite{9402750}. 
	
	The Alice-to-IRS direct path loss factor is ${\beta_{_{\text{{AR}}}}} \triangleq  \pmb{G}_{\text{\text{A}}}+\pmb{G}_{\text{IRS}}-35.9-22 \log_{10}(d_{_{\text{{AR}}}})$ in dB \cite{BOL19,kammoun2020asymptotic}, where $d_{_{\text{{AR}}}}$ is the distance between Alice and the IRS in meters, $\pmb{G}_{\text{\text{A}}}$ is Alice antenna gain, and $\pmb{G}_{\text{IRS}}$ is the IRS elements' antenna gain. The path loss factor from the IRS to the users and the Eve is ${\beta_{\text{R}i}} \triangleq  \pmb{G}_{\text{IRS}}-33.05-30 \log_{10}(d_{\text{R}i})$ dB \cite{BOL19,kammoun2020asymptotic}, for $i\in\{k,e\}$, $d_{\text{R}i}$ is the distance between the IRS and the users and Eve in meters. 
	The spatial correlation matrix is $[\pmb{R}_{Ri}]_{l,\bar{l}} \triangleq  \exp{(j\pi (l-\bar{l}) \sin \hat{\vartheta} \sin \hat{\aleph})}$ for $i\in \{k,e\}$, where $\hat{\aleph}$ is the elevation angle and $\hat{\vartheta}$ is the azimuth angle \cite{kammoun2020asymptotic}. The elements of the Alice-to-IRS channel are generated by $[\pmb{G}_{_{\text{{AR}}}}]_{a,b} \triangleq  \exp({j\pi \left((b-1) \sin\overline{\Theta}_b \sin \overline{\vartheta}_b+(a-1) \sin(\Theta_n) \sin(\vartheta_b)\right)})$, where ${\Theta}_n \in (0,2 \pi)$, ${\vartheta}_n \in (0,2 \pi)$, and $\overline{\Theta}_n \triangleq  \pi-\Theta_n$, and $\overline{\vartheta}_n \triangleq  \pi+\vartheta_n$ \cite{kammoun2020asymptotic}. The small-scale fading channel gain $\widehat{\pmb{u}}_{i}$ for $ i \in \{k,e\}$ is modeled as a Rician fading channel with K-factor$=3$.
	
	We define the CSI error bounds as $\Omega_{i} \triangleq \delta_i \|\widehat{\pmb{u}}_i\|_2, \forall k$, where $\delta_i \in [0,1), i\in \{k,e\}$ is the relative amount of CSI uncertainty. When $\delta=0$, Alice can obtain the perfect CSI of the IRS to users/Eve reflected channel. In the FBR, we set $\varsigma_{\max}=\varsigma_{e}^{\max}=10^{-5}$. The packet length $N_t$ is defined by the transmission duration and the bandwidth, as previously defined. Hence, we set the transmission duration as $T=0.1$ ms, which is suitable for FBR transmission \cite{ben18}. The choice of $0.1$ ms end-to-end delay ensures having a quasi-static channel during FBR communication \cite{sh17cross}. Unless stated otherwise, the simulation parameters are defined in Table \ref{T2}. We multiply the results by $\log_2(e)$ to convert them to bps/Hz. Finally, we compare the proposed imperfect CSI framework with the perfect CSI case as an FBR baseline. The perfect CSI formulation is solved using a step-descent method combined with a direct search procedure for handling the unit modulus constraint, whereas the imperfect CSI case relies on an LMI-SCA based beamforming update together with a PCCP-based IRS phase optimization. This comparison provides a fair, consistent, and FBR-compliant evaluation of performance under identical secrecy models and system assumptions.
    % \textcolor{black}{Finally, we compare the results of this study with those in our previous work \cite{10437125826}, which addressed the FBR scenario under perfect CSI using a step descent algorithm. In contrast, the proposed approach employs a SCA, $S$-procedure framework combined with the PCCP algorithm, enabling robust optimization under imperfect CSI conditions.}
    % \textcolor{black}{Finally, given \cite{10437125826}, which investigated the FBR scenario under perfect CSI using a gradient step-descent algorithm, we compare the results obtained by our proposed algorithm with those in \cite{10437125826}. Unlike \cite{10437125826}, our proposed algorithm adopts an SCA, and $S$-procedure-based framework integrated with the PCCP algorithm, enabling robust optimization under imperfect CSI conditions.}

	\subsection{Performance Evaluation}
	\begin{figure}[t!]
		\centering
		\includegraphics[width=0.5\textwidth]{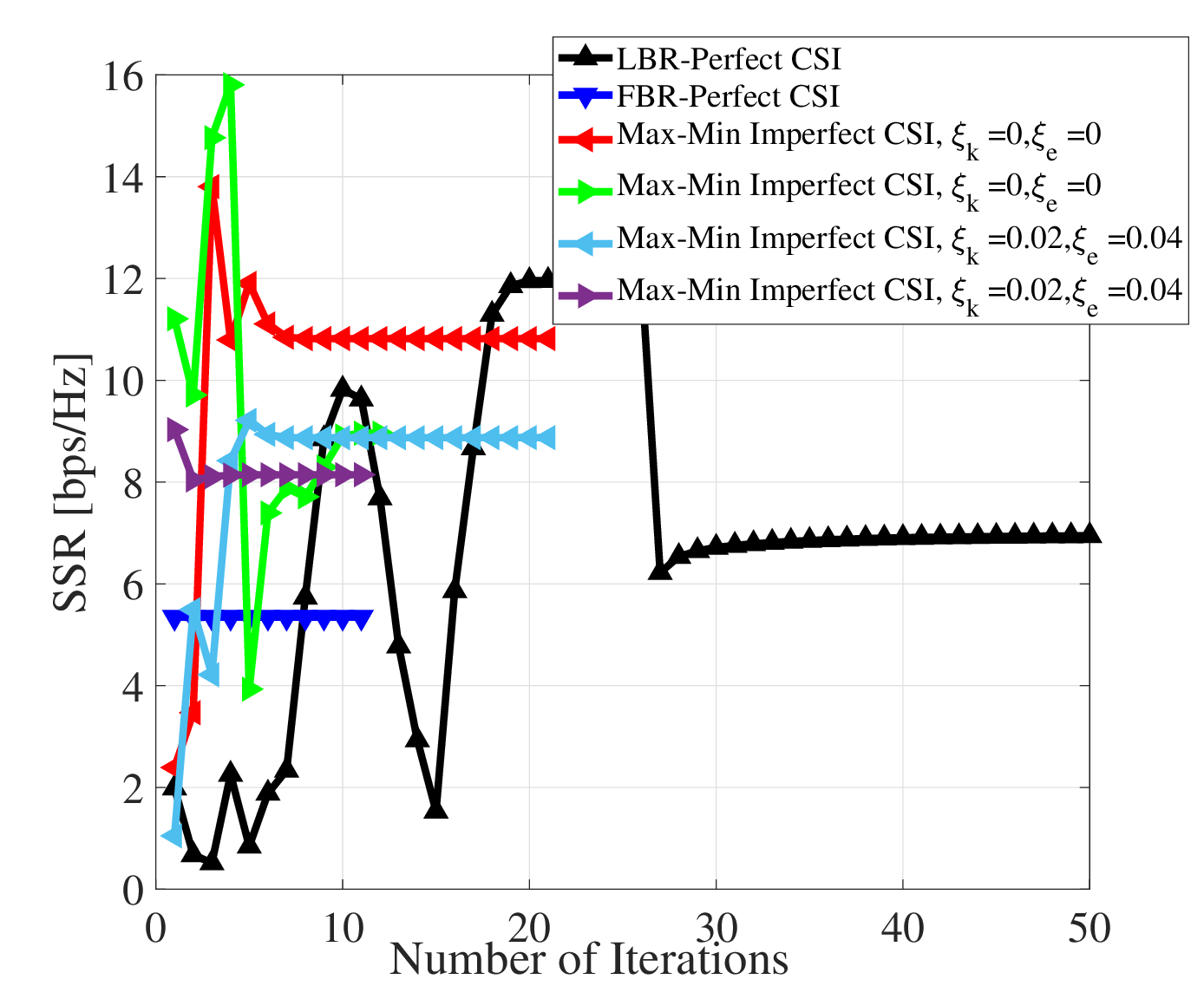}
		\caption{Convergence rate of the SSR algorithms under perfect and Imperfect CSI with $M=10$, $\mathcal{K}=6$, $N=16$.}
		\label{Fig2}
	\end{figure}
	
	% Fig. \ref{Fig2}  illustrates the convergence rate of the SSR algorithms under perfect CSI and Imperfect CSI for the FBR and the LBR systems. \textcolor{black}{It can be observed that the proposed FBR Algorithm \ref{alg4} requires fewer iterations to achieve convergence compared to the perfect CSI-based SSR maximization algorithm %in \cite{10437125826}.
 %    This result demonstrates the enhanced efficiency of the proposed approach, which achieves faster convergence while preserving solution optimality.} Furthermore, the FBR algorithm needs fewer iterations to converge since the local optimal solution obtained in the LBR case is used as the initial feasible point for the FBR counterparts. Such a choice highlights the importance of choosing the best initial feasible points in the FBR case. 
    Fig. \ref{Fig2} illustrates the convergence rate of the SSR algorithms under perfect CSI and Imperfect CSI for the FBR and the LBR systems. It can be observed that the proposed FBR Algorithm \ref{alg4} requires fewer iterations to achieve convergence compared to the perfect CSI-based SSR maximization algorithm. 
This result demonstrates the enhanced efficiency of the proposed approach, which achieves faster convergence while preserving solution optimality. Moreover, the fast convergence reduces the computational complexity associated with updating the beamforming vectors and IRS phase shifts, thereby enabling the optimization procedure to be completed within the stringent latency requirements of URLLC under quasi-static channel conditions. Furthermore, the FBR algorithm needs fewer iterations to converge since the local optimal solution obtained in the LBR case is used as the initial feasible point for the FBR counterparts. Such a choice highlights the importance of choosing the best initial feasible points in the FBR case.
	
	One way to evaluate the system performance in the proposed algorithms is by using the arithmetic mean. The arithmetic mean is defined as, $\text{arithmetic mean}= \frac{1}{\mathcal{K}}{\sum}_{k=1}^{\mathcal{K}}\mathcal{S}_{k}^{\mathcal{F}}(\pmw,\btheta)$. Fig. \ref{Fig3} plots the achieved arithmetic mean while varying the number of Alice's antennas $M$, with $\mathcal{K}=7$, $N=10$, and $\delta_k=\delta_e=0.02$. As expected, the proposed SSR algorithms' achieved arithmetic mean increases with $M$, since increasing $M$ provides more degrees of freedom, hence it increases the users' SINR. The SSR-FBR algorithms achieve a lower arithmetic mean than their LBR-SSR counterpart due to the FBR constraints, namely, the transmission duration and the latency.

	\begin{figure}[t!]
		\centering
		\includegraphics[width=0.5\textwidth]{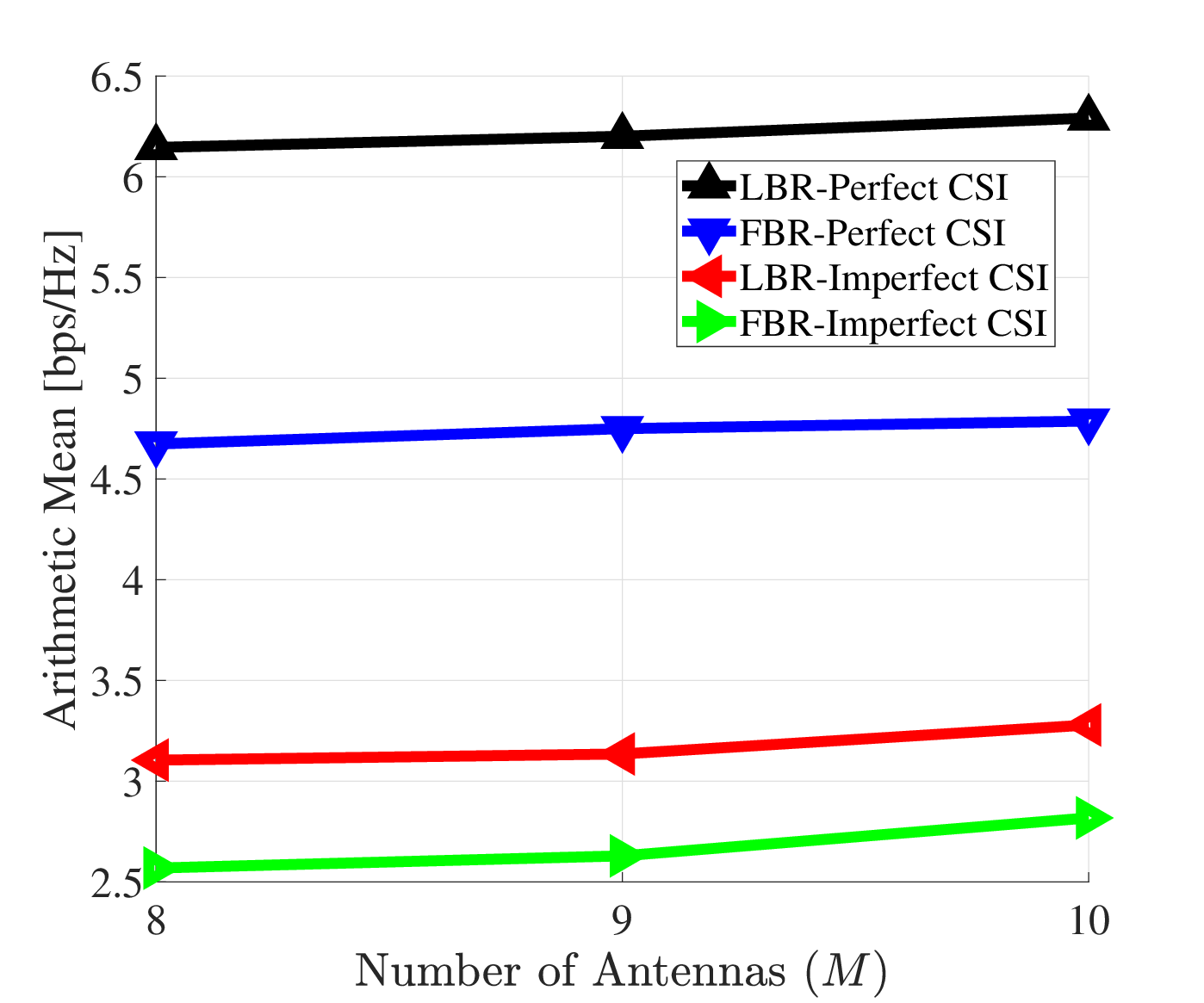}
		\caption{SR Arithmetic mean Vs the number of Alice's antennas $M$ with $ \mathcal{K}=7, N=10$, and $\delta_k=\delta_e=0.02$.}
		\label{Fig3}
	\end{figure}
	% \begin{figure}[t!] 
	% 	\centering
	% 	\includegraphics[width=0.5\textwidth]{Figures/Numericalresults/ICSI/AMloopkcombined.eps}
	% 	\caption{SR Arithmetic mean against the number of users $\mathcal{K}$ with $M=10$, $N=16$, and $\delta_k=\delta_e=0.02$.}
	% 	\label{Fig4}
	% \end{figure}
	\begin{figure}[t!] 
		\centering
		\includegraphics[width=0.5\textwidth]{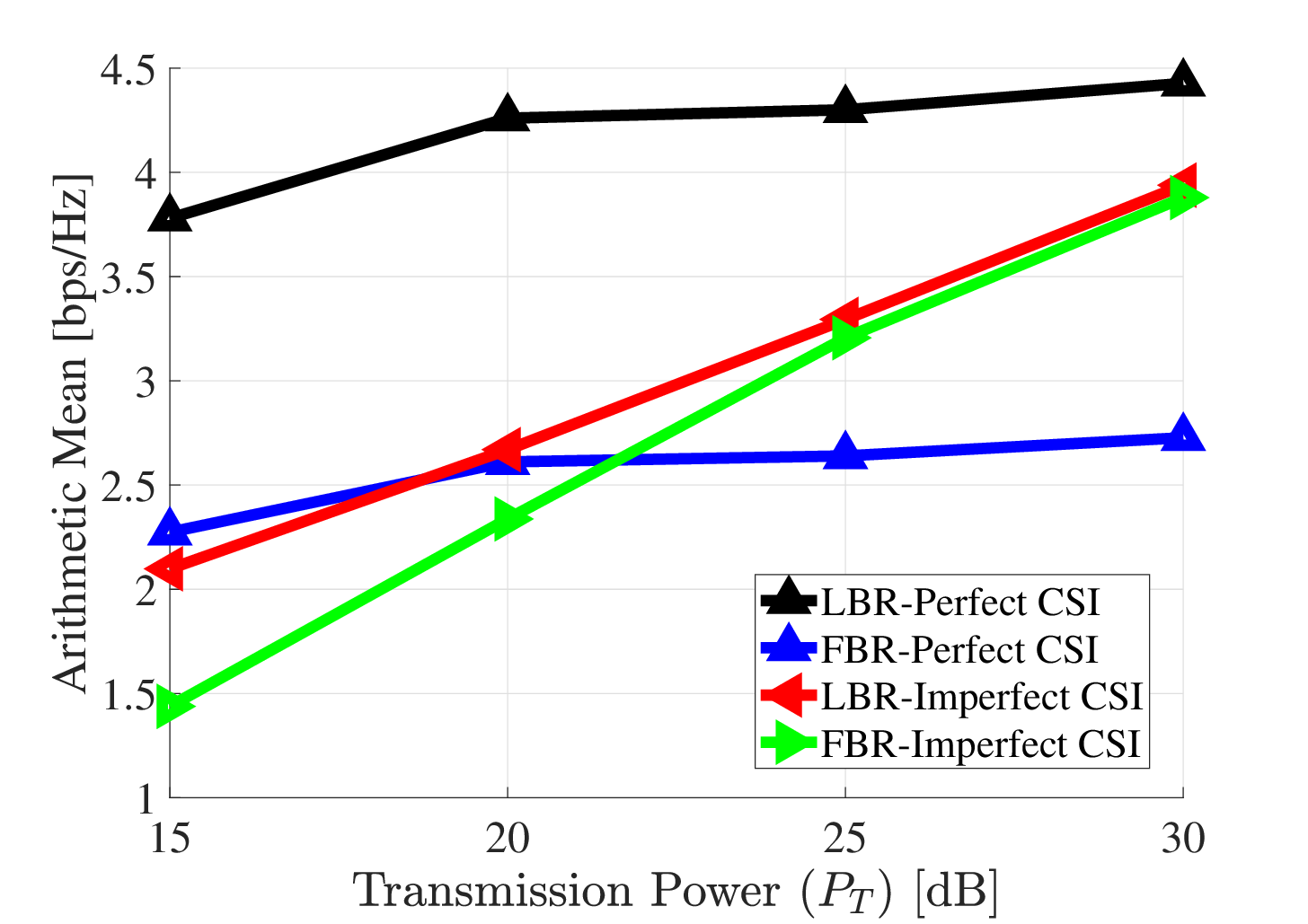}
		\caption{SR Arithmetic mean against the transmission power budget $P_T$ with $M=10$, $\mathcal{K}=7$, $N=16$, and $\delta_k=\delta_e=0.02$.}
		\label{Fig5}
	\end{figure}
	\begin{figure}[t!] 
		\centering
		\includegraphics[width=0.5\textwidth]{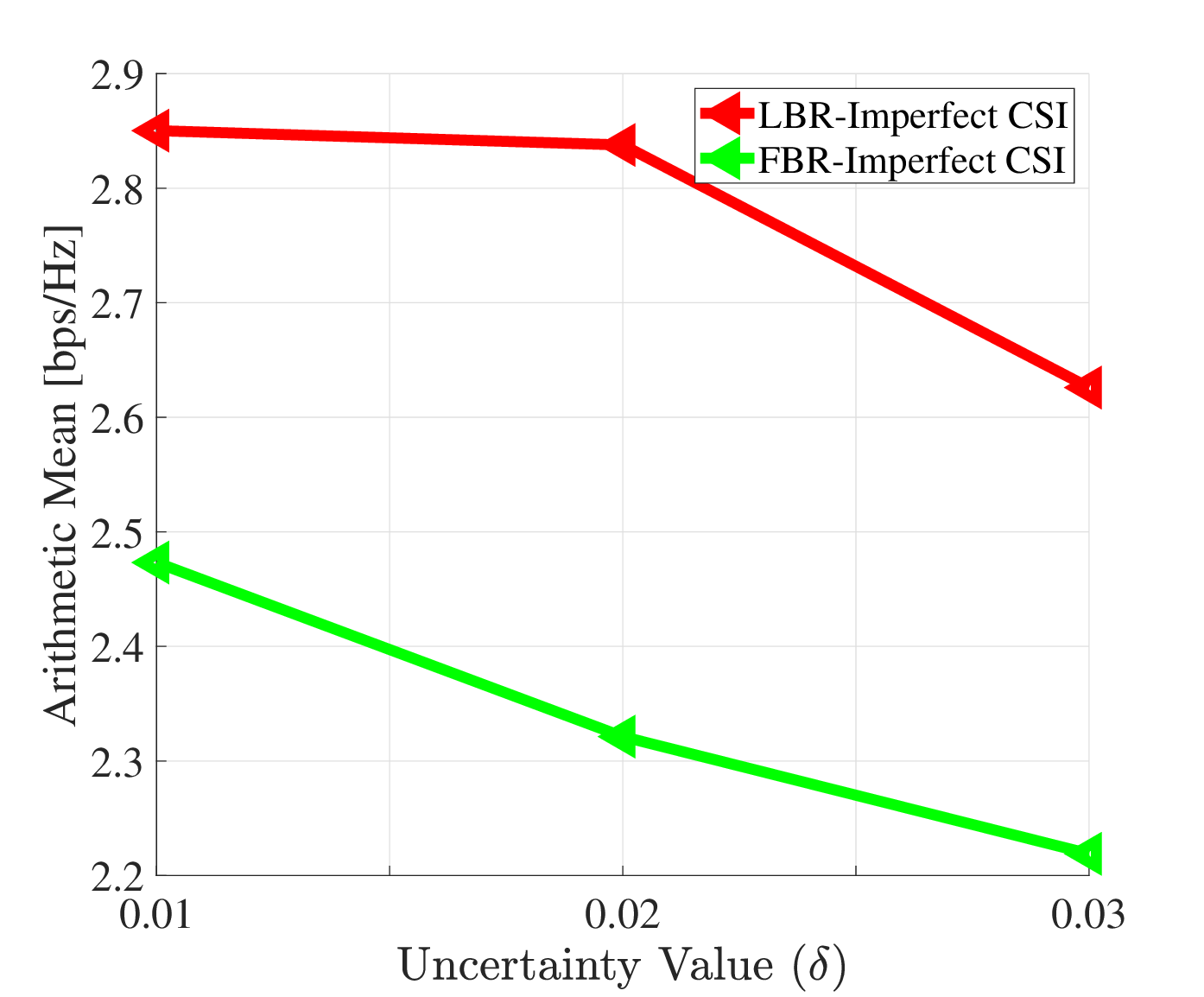}
		\caption{SR Arithmetic mean Vs. the uncertainty level $\delta_k=\delta_e$ with $M=10$, $\mathcal{K}=7$, and $N=10$.}
		\label{Fig6}
	\end{figure}
	% \begin{figure}[t!] 
	% 	\centering
	% 	\includegraphics[width=0.5\textwidth]{Figures/Numericalresults/ICSI/minSRcombinedloopt2.eps}
	% 	\caption{user's min SR Vs. Transmission Duration $T$ with $M=10$, $\mathcal{K}=6$ $N=16$, and $\delta_k=\delta_e=0.02$.}
	% 	\label{Fig7}
	% \end{figure}
	\begin{figure}[t!] 
		\centering
		\includegraphics[width=0.5\textwidth]{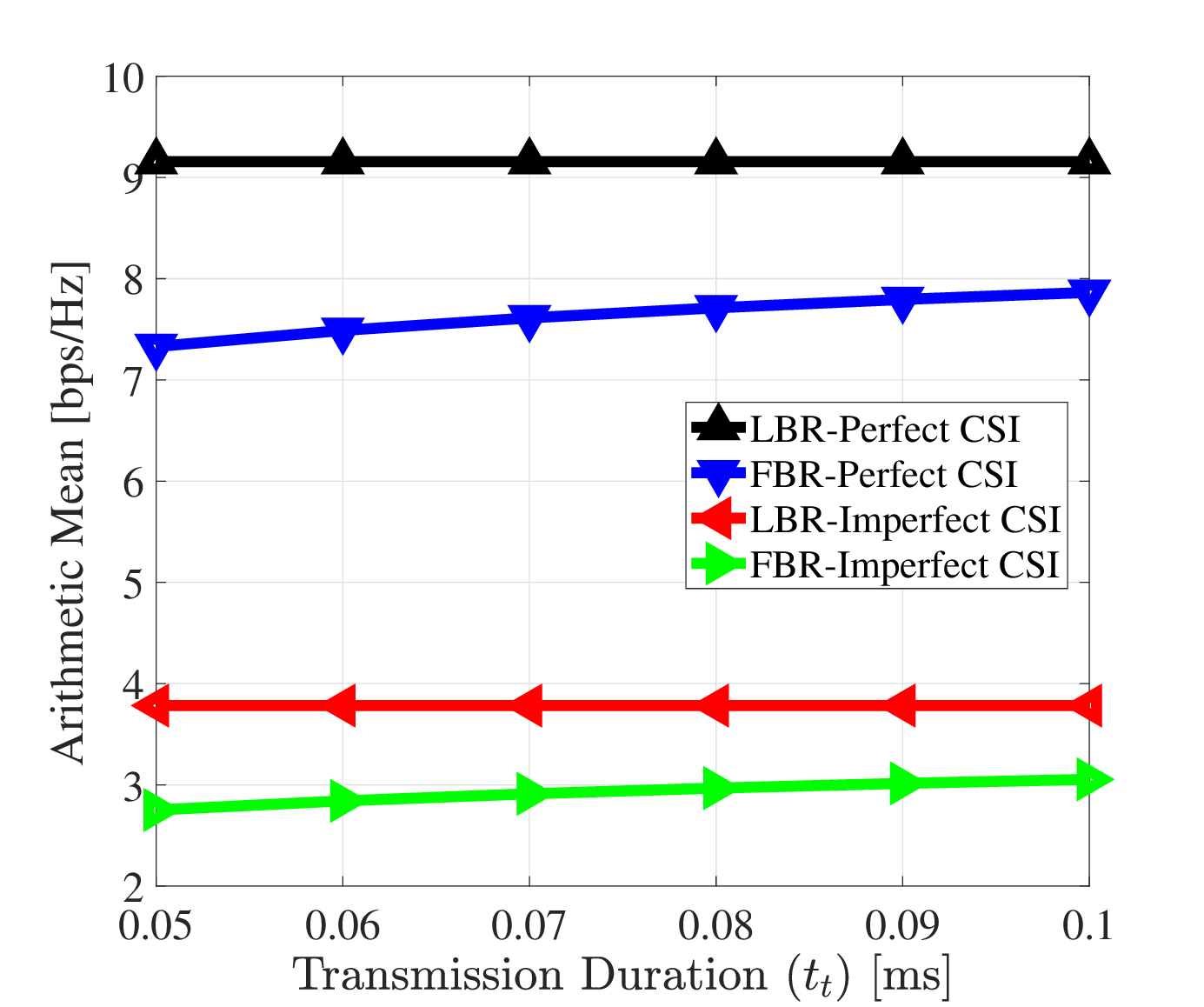}
		\caption{SR Arithmetic mean Vs. Transmission Duration $T$ with $M=10$, $\mathcal{K}=6$ $N=16$, and $\delta_k=\delta_e=0.02$.}
		\label{Fig8}
	\end{figure}
	
	% Fig. \ref{Fig4} portrays the arithmetic mean against the number of users $\mathcal{K}$, with $M=10$, $N=16$, and $\delta_k=\delta_e=0.02$. From this figure, one can notice that the achieved arithmetic mean obtained by all algorithms decreases with $\mathcal{K}$. This is mainly due to the fact that the inter-user interference increases with  $\mathcal{K}$, hence the arithmetic mean decreases. In addition, the FBR counterparts suffer from the same effect as the LBR parts, since the FBR initial feasible point is generated from the optimal solution obtained by solving its LBR counterparts. It can be noticed that the imperfect CSI SSR algorithms achieve a lower arithmetic mean than their perfect CSI counterparts. This is because the imperfect CSI lowers the users' SINR, thus the arithmetic mean tends to decrease. 

    \begin{figure}[t!] 
		\centering
		\includegraphics[width=0.5\textwidth]{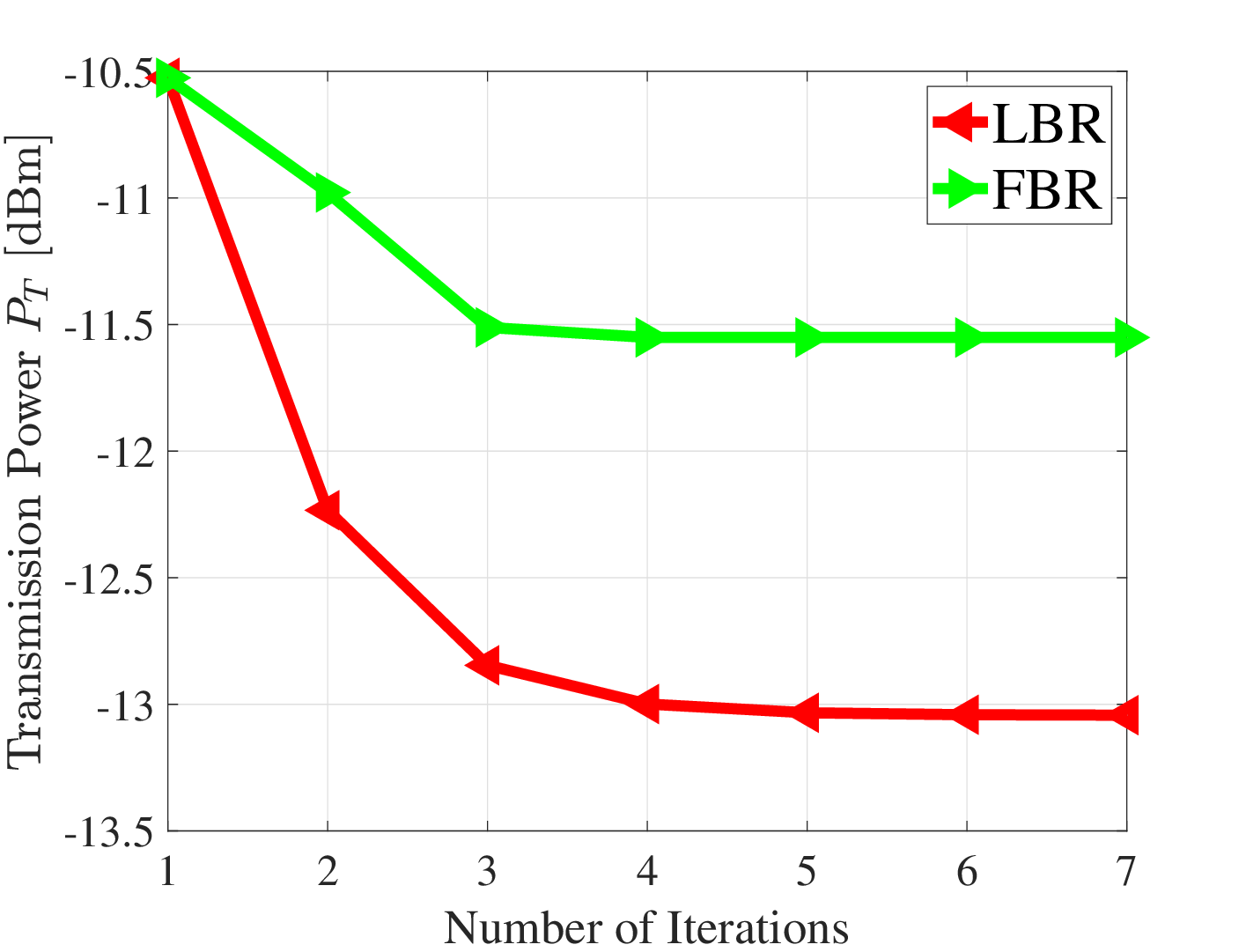}
		\caption{Convergence rate of the power minimization algorithm with $M=10$, $\mathcal{K}=6$ $N=20$, and $\delta_k=\delta_e=0.02$.}
		\label{Fig9}
	\end{figure}
    \begin{figure}[t!] 
		\centering
		\includegraphics[width=0.5\textwidth]{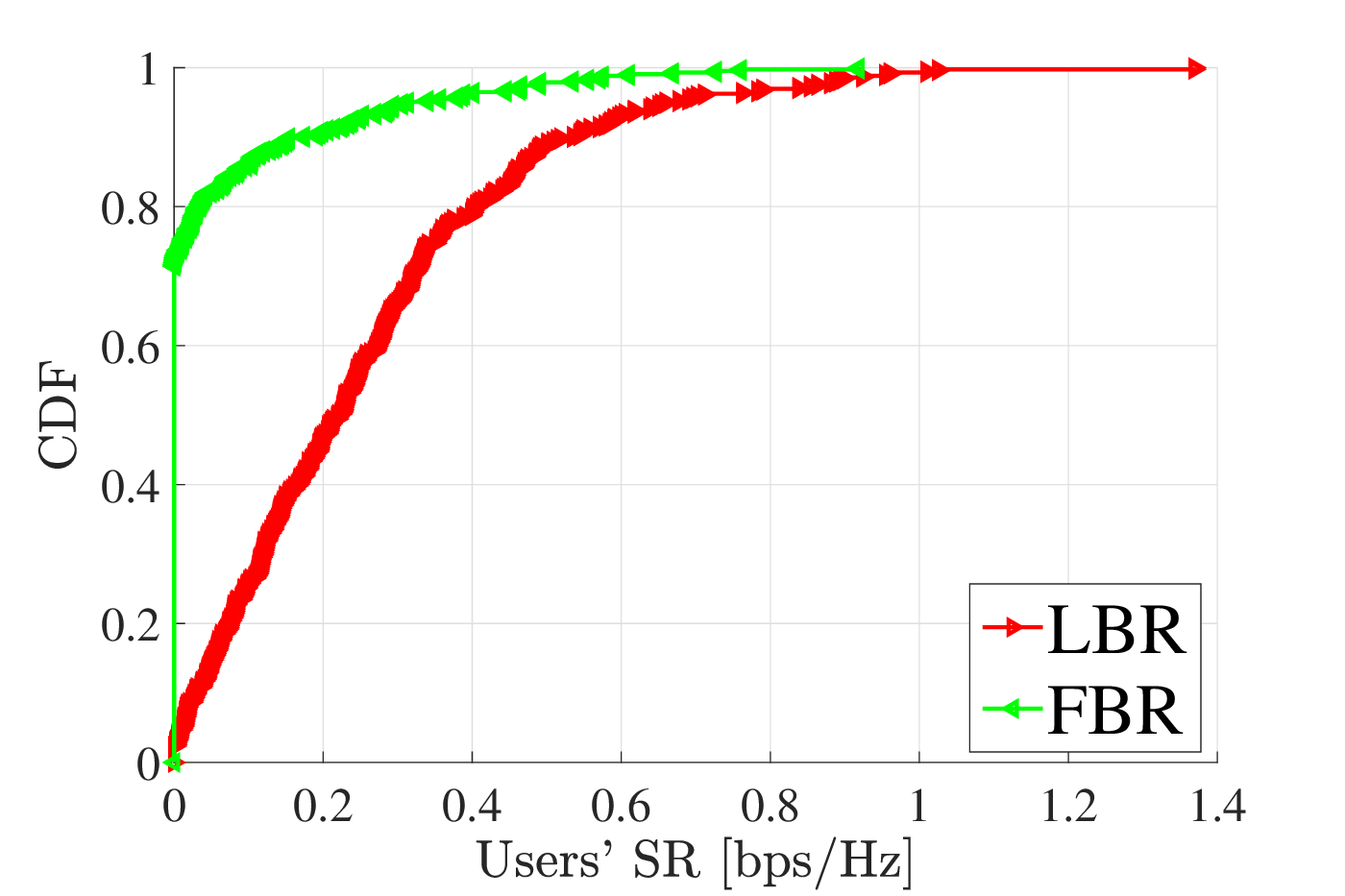}
		\caption{SR Arithmetic mean Vs. Transmission Duration  with $M=10$, $\mathcal{K}=5$ $N=20$, and $\delta_k=\delta_e=0.02$.}
		\label{Fig10}
	\end{figure}

	Fig. \ref{Fig5} illustrates the arithmetic mean against the transmission power budget $P$ with $M=10$, $\mathcal{K}=7$, $N=16$, and $\delta_k=\delta_e=0.02$. The achieved arithmetic mean obtained by all algorithms increases with $P$. As we can see, increasing $P$ improves the users' SINR; hence, the arithmetic mean increases. The FBR algorithms illustrate a similar behavior to that of their LBR counterparts, as expected.

	The arithmetic mean is illustrated against the relative amount of CSI uncertainty $\delta$ in Fig. \ref{Fig6}. The achieved arithmetic mean in the LBR and the FBR algorithm decreases with $\delta$. This is expected since, as the uncertainty increases, the users' SINR decreases.  However, even under high uncertainty, the proposed algorithms achieve system secrecy, even under the strict constraints of the FBR case. The obtained results demonstrate the robustness of the proposed LBR/FBR algorithm.

	% The minimum SR and the SSR arithmetic mean are illustrated against the transmission duration $T$ in Fig. \ref{Fig7} and Fig. \ref{Fig8}, respectively. Both the minimum SR and the SSR arithmetic mean increase with $T$. Even at low $T$, our proposed algorithms can provide system secrecy, which demonstrates the advantages of our algorithms. Finally, it can be noticed that the achieved LBR-SR serves as an upper bound for the SR-FBR counterpart, since the LBR assumes that the transmission duration $T = \infty$.

    The  SSR arithmetic mean is illustrated against the transmission duration $T$ in Fig. \ref{Fig8}. The SSR arithmetic mean increase with $T$. Even at low $T$, our proposed algorithms can provide system secrecy, which demonstrates the advantages of our algorithms. Finally, it can be noticed that the achieved LBR-SR serves as an upper bound for the SR-FBR counterpart, since the LBR assumes that the transmission duration $T = \infty$.

    Fig. \ref{Fig9} illustrates Algorithm \ref{alg7} convergence rate under both the FBR and the LBR. The proposed algorithm exhibits a stable and monotonic convergence pattern in both regimes, confirming its numerical robustness and effectiveness. It can be observed that the FBR case consistently converges to a higher transmit power level compared to its LBR counterpart. This is expected, as in the FBR regime the users'/Eve's dispersion factors introduce additional secrecy constraints in the optimization process, which consequently increases the required transmit power to ensure secrecy.

        To evaluate the distribution of users' secrecy rates, we illustrate the empirical cumulative distribution function (CDF) of the SR
    $\widehat{\mathcal{S}}_{k}^{\mathcal{F}}(\mathbf{w}, \mathbf{\theta})$ in Fig.~\ref{Fig10} for the LBR and the FBR. The empirical CDF is defined as:
    \begin{align}
        F_{\widehat{\mathcal{S}}_k^{\mathcal{F}}}(r) 
        &\triangleq \Pr\Big(\widehat{\mathcal{S}}_{k}^{\mathcal{F}}(\mathbf{w}, \mathbf{\theta}) \le r \Big), \\
        &\triangleq \frac{1}{N_s} \sum_{i=1}^{N_s} 
        \mathbf{1}\Big\{ \widehat{\mathcal{S}}_{k}^{\mathcal{F},(i)}(\mathbf{w}, \mathbf{\theta}) \le r \Big\}, 
        \quad r \in [0, \widehat{\mathcal{S}}_k^{\mathcal{F},\max}],
    \end{align}
    where $\widehat{\mathcal{S}}_{k}^{\mathcal{F},(i)}(\mathbf{w}, \mathbf{\theta})$ denotes the SR of user $k$ in the $i$-th simulation trial, $N_s$ is the total number of realizations, $\mathbf{1}\{\cdot\}$ is the indicator function, and $\widehat{\mathcal{S}}_k^{\mathcal{F},\max}$ is the maximum achievable per-user SR. It can be seen that the LBR CDF consistently lies to the right of the FBR CDF, reflecting higher system's SR across all realizations. This is expected, as the LBR regime is not constrained by finite blocklength or strict latency requirements. In contrast, the FBR regime must satisfy stringent blocklength and latency constraints, which necessitates prioritizing users with favorable channel conditions. To meet these constraints, the algorithm adaptively concentrates the available $P_T$ on users with favorable channel realizations. This strategic allocation ensures the system-level secrecy, demonstrating the algorithm's robustness and effectiveness under strict FBR conditions.
    \begin{figure*}[!t] 
		\normalsize 
		\begin{align}
			&\ln\left|I_n+[\pmb{\Lambda}]^2(\pmb{\digamma})^{-1}\right|\geq \ln\left|I_n+[\hat{\Lambda}]^2(\hat{\digamma})^{-1}\right|
			-\la [\hat{\Lambda}]^2(\hat{\digamma})^{-1}\ra+2\Re\{\la \hat{\Lambda}^H(\hat{\digamma})^{-1}\mathbf{A}\ra\}
			-\la (\hat{\digamma})^{-1}-(\hat{\digamma}+[\hat{\Lambda}]^2)^{-1},[\mathbf{A}]^2+\pmb{\digamma}\ra , \tag{80} \label{fund1}  \\
			&\ln(1+\sum_{i=1}^{l} |{z}_i|^2) \geq \ln(1+\sum_{i=1}^{l} |\bar{z}_i|^2)-\sum_{i=1}^{l} |\bar{z}_i|^2 +\sum_{i=1}^{l} 2 \Re\{\bar{z}_i^* {z}_i\} - \frac{\sum_{i=1}^{l} |\bar{z}_i|^2 \left(1+\sum_{i=1}^{l} |{z}_i|^2\right)}{1+\sum_{i=1}^{l} |\bar{z}_i|^2}.  \tag{81} \label{fund2}
		\end{align} 
		\hrulefill 
	\end{figure*}
	\section{Conclusions} \label{conc}
In this paper, we proposed a framework to achieve system secure communications under the FBR constraints for IoT settings in URLLC applications aided by an IRS. Specifically, by employing linearization and various non-convex optimization techniques, we developed computationally efficient algorithms to maximize the sum secrecy rate (SSR) under perfect and imperfect CSI from the IRS to the users, as well as perfect, imperfect, and unknown CSI from the IRS to the eavesdropper. Extensive simulation results demonstrated that the proposed SSR maximization algorithm ensures secure communication under FBR constraints even with imperfect CSI, while the power minimization algorithm maintains system secrecy when Eve's CSI is unknown or unavailable. 

The confidentiality of the proposed design is ensured through the explicit incorporation of the FBR-SR model, which accounts for both decoding error probability and information leakage. By jointly optimizing the beamforming vectors and IRS reflective coefficients, the proposed algorithms minimize the eavesdropper's ability to decode confidential messages while satisfying users' QoS constraints. Furthermore, the residual transmit power is strategically utilized as AN in scenarios where Eve's CSI is unknown, further degrading Eve's SINR and guaranteeing information-theoretic secrecy across all CSI conditions. These mechanisms collectively ensure confidentiality within the FBR regime, where short packet lengths and latency constraints pose significant challenges compared to traditional LBR-based designs.
It is worth noting that if Eve is positioned too close to a legitimate user, guaranteeing secrecy for all users becomes infeasible. In such cases, additional countermeasures such as cryptography-based or friendly jamming schemes may be required. Future work could extend the framework to \emph{multi-IRS-aided} networks with cooperative phase-shift optimization and coordinated beamforming~\cite{vmimo}. Furthermore, testing under spatially correlated channels, as encountered in urban IRS-user links, would provide a more realistic evaluation of secrecy-rate performance and prevent overestimation of gains. Overall, the proposed methods offer a scalable and robust foundation for secure IRS-aided URLLC networks under practical CSI imperfections and FBR constraints.

	\begin{appendices}
		\section{Proof of theorem 1} \label{AppB}
        
        \textcolor{black}{We first verify that the surrogate function used in problems $(\mathcal{P}1.3)$ 
and $(\mathcal{P}1.4)$ is a valid SCA surrogate for the original FBR-SR function. 
The surrogate $\widetilde{\mathcal{S}}_{k}^{\mathcal{F}}(\mathbf{w},\boldsymbol{\theta})$ 
in \eqref{SRFBRtot} is obtained by applying first-order Taylor expansions to the smooth 
functions $C_k$, $C_e$, and $V_i$ around the $\iota$ iterate 
$(\mathbf{w}^{(\iota)},\boldsymbol{\theta}^{(\iota)})$. Since 
$\gamma_i(\mathbf{w},\boldsymbol{\theta})$ is quadratic in $\mathbf{w}$ and 
linear in $\boldsymbol{\theta}$, all three functions are continuously differentiable and have Lipschitz-continuous gradients \cite{scutari2014parallel}. Hence, for a finite constant $L$, the 
linearization error of the dispersion term satisfies
\begin{align}
\big| V_i(\mathbf{w},\boldsymbol{\theta}) -
\hat{V}_i(\mathbf{w},\boldsymbol{\theta}\mid 
\mathbf{w}^{(\iota)},\boldsymbol{\theta}^{(\iota)}) \big|
\le L \|(\mathbf{w},\boldsymbol{\theta}) - 
(\mathbf{w}^{(\iota)},\boldsymbol{\theta}^{(\iota)})\|^2,
\end{align}
Moreover, the surrogate satisfies the standard 
SCA conditions:
\begin{align}
\widetilde{\mathcal{S}}_{k}^{\mathcal{F}}(\mathbf{w},\boldsymbol{\theta})
\le \mathcal{S}_{k}^{\mathcal{F}}(\mathbf{w},\boldsymbol{\theta}),
\end{align}
and
\begin{align}
\nabla \widetilde{\mathcal{S}}_{k}^{\mathcal{F}}(\mathbf{w}^{(\iota)},
\boldsymbol{\theta}^{(\iota)})
= \nabla \mathcal{S}_{k}^{\mathcal{F}}(\mathbf{w}^{(\iota)},
\boldsymbol{\theta}^{(\iota)}).
\end{align}
Thus, \eqref{SRFBRtot} is a valid SCA surrogate for the original FBR-SR function.
}    

        \textcolor{black}{Next, to prove that the solution obtained by Algorithm \ref{alg1} is a locally optimal solution for problem $(\mathcal{P}1)$, we} begin by demonstrating that the sequence $R_s(\wko, \thetako)$ is non-decreasing, i.e., $R_s(\wko, \thetako) \geq R_s(\wk, \thetak)$ for all $\iota > 0$. To prove this, we start by computing $\wko$ by solving problem $(\mathcal{P}1.3)$ with the aid of the CVX solver. Since the CVX solver guarantees the optimal solution, we can conclude that $R_s(\wko, \pmt) \geq R_s(\wk, \btheta)$ holds. Subsequently, by solving problem $(\mathcal{P}1.4)$, we compute $\thetako$, where it is again true that $R_s(\pmw, \thetako) \geq R_s(\pmw, \thetak)$, as previously established.
		
        Thus, by combining the results from solving problems $(\mathcal{P}1.3)$ and $(\mathcal{P}1.4)$, we obtain the following inequality:
        \begin{align} \label{convxtot}
    R_s(\wko, \thetako) \geq R_s(\wk, \thetak),
    \end{align}
    which demonstrates that the optimal sequence $\{(\wko, \thetako)\}$ converges to the point $\{({\pmw}^*, {\btheta}^*)\}$, which represents the solution obtained from solving both $(\mathcal{P}1.3)$ and $(\mathcal{P}1.4)$.

Next, we proceed to establish that the converged point $\hat{X}^* \triangleq \{\pmw^*, \btheta^*\}$ is a locally optimal solution to problem $(\mathcal{P}1)$. To achieve this, we demonstrate that the converged point satisfies the Karush-Kuhn-Tucker (KKT) conditions for the problem. The KKT condition for problem $(\mathcal{P}1.4)$ is satisfied at $\btheta^*$. Let $\pmb{Y}(\btheta)$ represent the objective function of problem $(\mathcal{P}1.4)$, and let $\pmb{T}({\pmb{X}}) = [\pmb{T}_1({\pmb{X}}), \pmb{T}_2({\pmb{X}}), \dots, \pmb{T}_I({\pmb{X}})]$ be the set of constraints for problem $(\mathcal{P}1)$. Then, the following holds:
\begin{align} \label{con1}
&\nabla_{\btheta^*} \pmb{Y}({\pmb{X}}^*) + \pmb{Z}^T \nabla_{\btheta^*} \pmb{T}({\pmb{X}}^*) = 0, \\
&z_i \geq 0, \, z_i \pmb{T}({\pmb{X}}^*) = 0, \forall i. \nonumber
\end{align}
where $\nabla_s$ denotes the gradient with respect to $s$, and $\pmb{Z} = [z_1, z_2, \dots, z_I]$ represents the set of optimal Lagrange multipliers. Similarly, the solution to problem $(\mathcal{P}1.3)$ is also locally optimal, and thus its KKT condition is satisfied with respect to $\pmb{W} = \pmw_k^*$. Specifically, we have:
\begin{align} \label{con2}
&\nabla_{\pmb{W}^*} \pmb{Y}({\pmb{X}}^*) + Z^T \nabla_{\pmb{W}^*} \pmb{T}({\pmb{X}}^*) = 0, \\
&z_i \geq 0, \, z_i \pmb{T}({\pmb{X}}^*) = 0, \forall i. \nonumber
\end{align}
By combining the conditions in \eqref{con1} and \eqref{con2}, we derive the following:
\begin{align}
&\nabla_{{\pmb{X}}^*} \pmb{Y}({\pmb{X}}^*) + Z^T \nabla_{{\pmb{X}}^*} \pmb{T}({\pmb{X}}^*) = 0, \\
&t_i \geq 0, \, z_i \pmb{T}({\pmb{X}}^*) = 0, \forall i, \nonumber
\end{align}
which corresponds to the KKT conditions for problem $(\mathcal{P}1)$ as presented in \eqref{P1-1}. Therefore, the converged point $\hat{\pmb{X}}^*$ is indeed the local optimal solution for problem $(\mathcal{P}1)$. $\blacksquare$
		\section{Proof of Lemma 1} \label{apendixB}
		Let $q$ be a scalar complex variable, and $a^{(\iota)}$ is the fixed point obtained at iteration $(\iota)$, then the following inequality holds \cite{9110587}:
		\begin{align}
			|a|^2 \geq a^{*,(\iota)} a +a^{*} a^{(\iota)} - a^{*,(\iota)} a^{(\iota)}.
		\end{align}
		By substituting $q$ with $({\pmb{u}}_k \pmb{\Phi} \pmb{G}_{_\text{AR}})\pmw_k$ we obtain \eqref{ICSIeqw}. Thus, this concludes the proof. $\blacksquare$
		\section{Proof of theorem 2} \label{AppD}
		Similar to the proof of Theorem \ref{theo1}, it can be shown the sequence $\widehat{R}_{s}(\wko,\thetako)$ is non-decreasing, i.e.  $\widehat{R}_{s}(\wko,\thetako) \geq \widehat{R}_{s}(\wk,\thetak)$ for all $\iota >0$, and the converged point is a locally optimal solution for problems. $\blacksquare$
		\section{Inequalities} \label{AppA}
		We begin by applying the inequality presented in (48) of \cite{TTN16}. Specifically, for any matrices $\mathbf{A}$ and $\pmb{\digamma}$, where $\hat{\Lambda}$ and $\hat{\digamma}$ are fixed points, the expression in (\ref{fund1}) holds.
		
		Next, we utilize Lemma 2 from \cite{niu2022joint}. In particular, for any $z_i$ where $i = 1, \dots, l$, and $\bar{z}_i$ is a fixed point, the inequality in (\ref{fund2}) is satisfied.
		
		Furthermore, the concave logarithmic function can be expressed as \cite[Eq. 15]{niu2022joint}: \setcounter{equation}{81}
		\begin{align} \label{fund3} 
			-\ln(1 + \pmb{\Upsilon}) \geq -\ln(1 + \bar{\Upsilon}) - \frac{1 + \pmb{\Upsilon}}{1 + \bar{\Upsilon}} + 1.
		\end{align}
		Additionally, the following inequality holds due to the concavity of the function $\sqrt{x}$ \cite{abughalwa2022finite}:
		\begin{equation}\label{xy}
			\sqrt{x} \leq \frac{\sqrt{\bar{x}}}{2} \left(1 + \frac{x}{\bar{x}}\right), ~ \forall x > 0, , \bar{x} > 0.
		\end{equation}
		Finally, for any vector $\mathbf{A} \in \mathbb{C}^n$ and $\bar{\mathbf{A}} \in \mathbb{C}^n$, along with scalars $B > 0, , \bar{B} > 0, , \sigma > 0$, the following inequality holds \cite{abughalwa2022finite}:
		\begin{align} \label{xt}
			\frac{|\mathbf{A}|^2}{B + \sigma} \geq \frac{|\bar{\mathbf{A}}|^2}{\bar{B} + \sigma} \left(2 \frac{\Re{\bar{\mathbf{A}}^H \mathbf{A}}}{|\bar{\mathbf{A}}|^2} - \frac{B + \sigma}{\bar{B} + \sigma} \right).
		\end{align}
	\end{appendices}
	\bibliographystyle{IEEEtran}
	\bibliography{surface}
\end{document}